\documentclass[sigconf]{acmart}

\usepackage{hyperref}
\definecolor{darkblue}{rgb}{0,0,0.5}
\hypersetup{
	pdftitle={TorPolice: Towards Enforcing Service-Defined Access Policies In Anonymous Systems},
	pdfauthor={Anonymous}
	pdfkeywords={Tor, anonymity, access control},
	colorlinks=true,
	urlcolor=darkblue,
	linkcolor=darkblue,
	citecolor=darkblue
}
\usepackage{lipsum}
\usepackage{setspace}
\usepackage{booktabs, tabularx}
\usepackage{array, xspace, subfigure, xcolor, amsmath, bm, array,amsfonts, balance,multirow,url}

\makeatletter
\newif\if@restonecol
\makeatother

\usepackage[linesnumbered,ruled,vlined]{algorithm2e}
\usepackage[noend]{algpseudocode}

\SetCommentSty{mycommfont}

\newcommand{\ea}{\emph{et al.}\xspace}

\newcommand{\ie}{\emph{i.e.,}\xspace}
\newcommand{\eg}{\emph{e.g.,}\xspace}
\newcommand{\sys}{TorPolice\xspace}

\newcommand{\first}{\textsf{(i)}\xspace}
\newcommand{\second}{\textsf{(ii)}\xspace}
\newcommand{\third}{\textsf{(iii)}\xspace}
\newcommand{\four}{\textsf{(iv)}\xspace}
\newcommand{\five}{\textsf{(v)}\xspace}

\usepackage{mathtools}
\DeclarePairedDelimiter\floor{\lfloor}{\rfloor}

\renewcommand\footnotetextcopyrightpermission[1]{} 
\newcommand{\paraspace}{\vspace{0.03in}}
\newcommand{\parab}[1]{\paraspace\noindent{\bf#1}}

\makeatletter
\def\@copyrightspace{\relax}
\makeatother

\usepackage[font={small,bf},skip=1.0pt]{caption}
\usepackage{amsthm,lipsum}
\makeatletter
\def\thm@space@setup{%
	\thm@preskip=2pt
	\thm@postskip=\thm@preskip 
}

\usepackage{siunitx}

\ccsdesc[500]{Security and privacy~Privacy-preserving protocols}
\ccsdesc[500]{Security and privacy~Security protocols}

\begin{document}
\title{\sys: Towards Enforcing Service-Defined Access Policies \\ in Anonymous Systems}
\titlenote{The initial version is published in IEEE ICNP 2017~\cite{icnpversion}, titled ``\sys: Towards Enforcing Service-Defined Access Policies for Anonymous Communication in the Tor Network''.}
\author{
	Zhuotao Liu$^*$,
	Yushan Liu$^\dag$,
	Philipp Winter$^\dag$,
	Prateek Mittal$^\dag$,
	Yih-Chun Hu$^*$
}
\affiliation{%
	\institution{$^*$ University of Illinois at Urbana-Champaign, $^\dag$ Princeton University}
}
\affiliation{%
	\institution{$^*$ \{zliu48,yihchun\}@illinois.edu, $^\dag$ \{yushan, pwinter, pmittal\}@princeton.edu}
}

\maketitle

\section{Abstract}
Tor is the most widely used anonymity network,
currently serving millions of users each day.
However, there is no access control in place for all these users,
leaving the network vulnerable to botnet abuse and attacks.
For example, criminals frequently use exit relays as stepping stones
for attacks, causing service providers to serve CAPTCHAs to
exit relay IP addresses or blacklisting them altogether,
which leads to severe usability issues for legitimate Tor users.
To address this problem, we propose \sys, the first
privacy-preserving access control framework for Tor.
\sys enables abuse-plagued service providers such as Yelp
to enforce access rules to police and throttle malicious requests coming from
Tor while still providing service to legitimate Tor users.
Further, \sys equips Tor with global access control for
relays, enhancing Tor's resilience to botnet abuse.
We show that \sys preserves the privacy of Tor users,
implement a prototype of \sys, and perform
extensive evaluations to validate our design goals. 
\section{Introduction}
In an era of mass surveillance, our online communications are being increasingly monitored 
by businesses and government entities to infer sensitive information. 
Technologies for anonymous communication aim to hide users' network identity (IP address) from 
untrusted destinations, as well as third parties on the Internet~\cite{tor,I2P,onion1,onion2}.  
Counting almost two million daily users, the Tor network~\cite{tor} is among the most
popular digital privacy tools. As of May 2017, the network consists
of over 7,000 volunteer-run relays, carrying nearly 100 Gbps of traffic~\cite{tor_metric}.
Tor clients\footnote{In this paper, we use the term client(s) to refer to the onion proxy (OP) software 
running on the Tor user's machine.} build a path (also known as Tor circuit) consisting 
of three relays (guard, middle and exit) to reach service providers such as Yelp or Wikipedia.
Tor is used by law enforcement, intelligence agencies, political dissidents, journalists, whistle-blowers, 
businesses, and ordinary citizens to enhance their online privacy~\cite{tor_users}. 

Today's Tor network does not implement any access control mechanism, meaning that anyone with a Tor client can use
the network without limitation. While the lack of access control fosters network growth, it has also caused various
problems, most importantly botnet abuse~\cite{hopper2013protecting}.  In practice, botnets
use Tor to attack third-party services, spam comment sections on websites, scrape content, and scan services for
vulnerabilities~\cite{cloudflare-trouble}.  In response, many service providers and content delivery networks (CDNs)
have started to treat Tor users as ``second-class" Web citizens~\cite{diff_treatment}, by either forcing Tor users to solve
numerous CAPTCHAs~\cite{cloudflare-trouble,tor-trouble} or blocking Tor exit relay IP addresses altogether. 

Another type of botnet-related abuse of Tor arises from command and control (C\&C) servers run as Tor onion
services (used to be known as hidden services)~\cite{botnet_abuse_1,botnet_abuse_2,botnet_abuse_3}. 
In the past, such events caused a rapid spike in the number
of Tor clients~\cite{cc_abuse, hidden_service_abuse}. Besides the reputational issue of Tor ``hosting'' botnet
infrastructure, the massive number of circuit creation requests from botnets is a heavy burden on Tor relays, causing 
significant performance degradation for legitimate Tor users (\eg frequent Tor circuit failures).  Other types of botnet
abuse include paralyzing Tor relays via relay flooding attacks~\cite{cellflood,sdos} and performing large-scale
traffic analysis via throughput or congestion fingerprinting~\cite{mittal:ccs11,murdoch:sp05}.

\parab{Contributions.} In this paper, we present \sys, the first privacy-preserving
access control framework for the Tor network.  Leveraging 
cryptographically computed \emph{network capabilities}, \sys enables  
service providers to define access policies for Tor connections, allowing them to 
throttle Tor-emitted abuse while still serving legitimate Tor users. Thus, \sys offers
a more viable alternative to abuse-plagued service providers than simply blocking all Tor connections.  
Further, \sys improves the Tor network's resilience to various botnet abuses by enabling global access control for Tor  relays. 
Crucially, \sys achieves these benefits while still retaining Tor's anonymity guarantees.

\sys's design introduces a set of fully distributed and partially trusted  
\emph{access authorities} (AAs) to manage and certify capabilities. 
To request capabilities from AAs, Tor clients must first obtain anonymous capability seeds 
which are types of resources that are costly to scale.  Both service providers and the Tor network 
provide differentiated service to Tor clients that possess valid capabilities so 
to enforce self-defined access rules. The AAs generate capabilities using blind signatures~\cite{blind}
to break the linkability between capability requesting and capability spending. 
We conduct a rigorous security analysis to prove that \sys does not weaken privacy 
guarantees offered by the current Tor network. 

We implement a prototype of \sys to demonstrate its practicality and evaluate the 
prototype extensively on our testbed, in the Shadow simulator~\cite{shadow}, via simulations and 
over the live Tor network. Our results show that \sys can effectively enforce 
service-selected access policies and mitigate large-scale botnet abuses against Tor at the cost
of negligible overhead.

We structure the rest of our paper as follows. We begin by outlining our problem
statement in \S~\ref{sec:problem_formulation}, followed by a design overview in
\S~\ref{sec:design_overview}.  The details behind access authorities are in
\S~\ref{sec:AAs}.  Next, \S~\ref{sec:serviceCapability} discusses how \sys can
control site access while \S~\ref{sec:TorCapability} discusses how we can
control access to the network itself. We analyze \sys's effect on Tor's
anonymity in \S~\ref{sec:security_analysis}, discuss its implementation in
\S~\ref{sec:implementation}, and evaluate it in \S~\ref{sec:evaluation}.
Finally, we present related work in \S~\ref{sec:related} and conclude our work
in \S~\ref{sec:conclude}. 
\section{Problem Formulation}\label{sec:problem_formulation}
In this section, we provide brief background on the Tor network (\S~\ref{sec:Tor_background}), outline \sys's design goals
(\S~\ref{sec:design_overview_goal}), and discuss our threat model (\S~\ref{sec:design_overview_threat}).

\subsection{Tor Background}\label{sec:Tor_background}
Tor clients anonymously connect to service providers (\eg WikiLeaks) by building three-hop circuits consisting of a guard, middle, 
and exit relay. Tor's use of layered encryption ensures that each relay only knows the identities of its direct neighbors 
(\ie the previous and next hop in the circuit). Clients randomly select these relays, weighted by the relays' bandwidth 
and their positions on the circuit. A list of all Tor relays---the network consensus---is published  
hourly by a set of nine globally-distributed directory authorities that are run  by volunteers trusted by the Tor Project.
While the directory authorities and guard relays learn a Tor client's network identity (\ie her IP address), 
they cannot observe the client's online activity. Exit relays, however, can monitor the client's 
activity, but do not know her identity.  Tor's anonymity stems from unlinking network identity from activity.

Besides client-side anonymity, Tor allows service providers to host their service anonymously
over Tor onion services (OS). Once an OS is set up, it creates circuits to at least three
relays severing as its \emph{introduction points} (IPs). Then, the OS publishes its \emph{descriptor}---which contains the IPs---to a 
distributed hash table that consists of a subset of all Tor relays. To connect to the OS, a Tor client first 
fetches the OS's descriptor using its \emph{onion address}, and then builds two 
circuits: one to an IP and another one to a randomly-selected relay called the 
\emph{rendezvous point} (RP). The client instructs the IP to send the identity 
of the RP to the OS, which then  creates a circuit to the RP to be able to finally communicate with the client.

\subsection{Design Goals}\label{sec:design_overview_goal}
\sys adds access control to the anonymous communication in Tor, 
benefiting both service providers and the Tor network.  
Different from prior capability based schemes~\cite{SIFF,tva,portcullis,netfence,MiddlePolice}, 
\sys's design needs to address a unique combination of the following three challenges: 
\first preserving Tor's anonymity guarantees, 
\second avoiding central points of control, and \third being incrementally deployable. 

\parab{Service-defined Access Policies.} 
Project Honey Pot lists nearly 70\% of all Tor exit relays as comment
spammers~\cite{cloudflare-trouble}, causing many service providers and CDNs to
block and filter traffic originating from the Tor network.  To reduce this
tension between Tor users and service providers, \sys must allow service
providers to define and enforce access rules for Tor connections, 
allowing them to throttle Tor-emitted abuse while still serving legitimate Tor users.
\sys is a flexible framework that allows service providers to define self-desired  access policies.  

\parab{Mitigate Botnet Abuse Against Tor.} 
Being a service provider itself, the Tor network is also subject to botnet abuse, such C\&C
servers hosted as onion services, and (D)DoS attacks against (selected) relays. \sys
allows the Tor network to control the network usage of Tor clients, making it
possible to throttle the abuse. In contrast to local rate limiting by each relay, 
\sys's access control mechanism is \emph{global}, meaning that an adversary cannot 
circumvent our defense by simply connecting to all relays. 

\parab{Preserving Tor User Privacy.}
\sys must not degrade Tor's anonymity guarantees.  While we add a new layer of
functionality to Tor (access control), this layer---like Tor itself---unlinks a client's identity
from its activity, and therefore preserves Tor users' online anonymity.

\parab{Fully Distributed and Partially Trusted Authorities.}  
In accordance with Tor's design philosophy of distributing trust, \sys relies on
a set of fully distributed and partially-trusted access authorities (AAs) to
manage capabilities.  An AA is operated either by the Tor Project, a
service provider, or a trusted third party.  Since Tor clients are free to
choose any AA to request capabilities, no single AA has a global view on all 
Tor clients.  Further, each AA is only partially trusted and a service provider 
can blacklist any misbehaving or compromised AA. 

\parab{Incrementally Deployable.} 
\sys must be incrementally deployable.  Up-to-date Tor clients, relays, and
service providers can benefit from a partially-deployed \sys immediately while
outdated entities can continue their operations.

\parab{Elided Design Goals.}
Various attacks seek to break Tor's unlinkability. For instance, an AS-level adversary may de-anonymize a Tor user's  
Internet activities if the adversary is in a position to monitor both ingress and egress traffic~\cite{raptor}.
\emph{\sys is not designed to mitigate those attacks on unlinkability}. Instead, we preserve the unlinkability 
guarantees that the Tor network currently provides.

\subsection{Adversary Model and Assumptions}
\label{sec:design_overview_threat}
We consider a Byzantine adversary that deviates from our protocol and abuses Tor in arbitrary ways.
The adversary can use Tor to abuse third-party services, \eg by scraping content, 
spamming comments, and scanning for vulnerabilities. The adversary may also 
abuse the Tor network directly, \eg by using Tor OSes as C\&C servers, performing traffic analysis,
or launching (D)DoS attacks against Tor relays. The adversary may further control  
a large number of bots, and hence a significant amount of resources.
The bots can act passively (\eg monitor Tor traffic) or actively 
(\eg spoof and manipulate packets). 

We assume that the AAs are well-connected to the Internet
backbone so that volumetric DDoS attacks against the whole set of AAs can be mitigated. 
Tor's existing directory authorities are subject to the same assumption. In practice, 
one way to assure this assumption is relying on DDoS prevention vendors~\cite{MiddlePolice}. 

\begin{figure}[t]
\centering
\mbox{
\subfigure{\includegraphics[width=\linewidth]{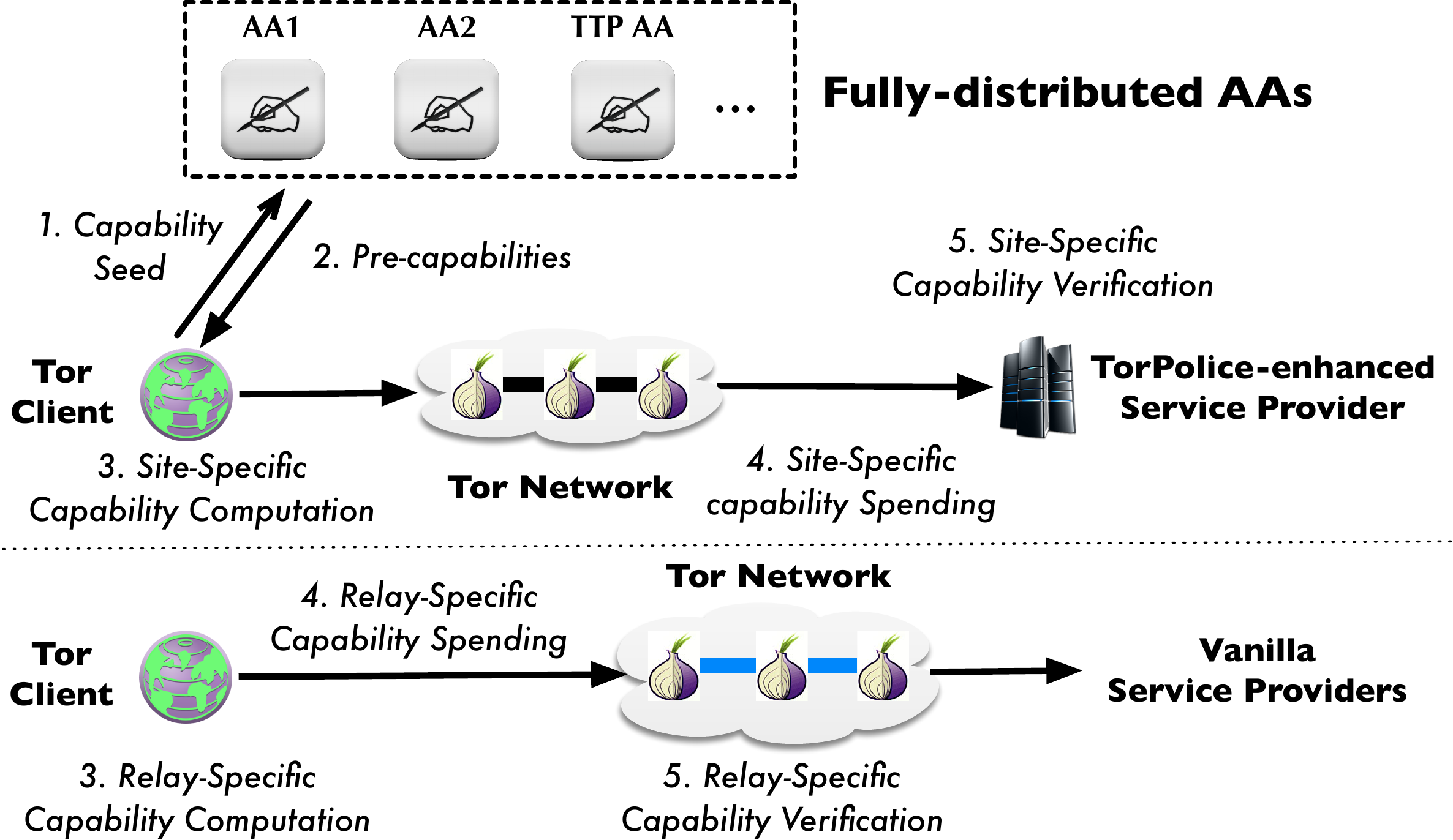}}}
\caption{The architecture of \sys. A Tor client (step 1) anonymously sends its capability seed 
to an AA to request pre-capabilities (step 2), based on which the client computes site(relay)-specific 
capabilities (step 3). The client then spends capabilities on either service access or Tor circuit creation (step 4). 
The capability recipients validate capabilities before allowing services (step 5). We intentionally  
separate two capability use cases for clear presentation.} \label{fig:sys_arch}
\end{figure}

\section{Design Overview} \label{sec:design_overview}
In a nutshell, \sys is a generic access control framework based on \emph{capabilities}. 
\sys enables both service providers and the Tor network itself to enforce 
access control on Tor clients to mitigate various types of botnet abuse 
caused by the lack of access control. To this end, we consider two types of capabilities: 
\emph{site-specific capabilities} for accessing \sys-enhanced service providers through Tor, and 
\emph{relay-specific capabilities} for creating  \sys-enhanced Tor circuits. 		
Both types of capabilities are signed by a set of fully-distributed Access Authorities (AAs) 
that are deployed either by the Tor Project, service providers, or trusted third parties. 
To request capabilities from a particular AA, a Tor client is required to possess a \emph{capability 
seed}---basically a costly-to-scale resource---accepted by the AA. Each AA accepts only a single 
type of capability seed. Since Tor clients are free to choose their AAs, no single AA has a global view on 
all Tor clients. \sys employs blind signatures~\cite{blind} to unlink the requesting and 
spending of capabilities. When requesting capabilities from an AA, Tor clients express 
what kind of capability they request because the issuing process for two capability types 
differs. An AA maintains separate signing keys and rate limiters for two capability types. 
 
Figure~\ref{fig:sys_arch} illustrates the capability requesting and spending process.
While both capability types have in common step one and two, the subsequent
steps differ.  A site-specific capability can only be spent at the service 
provider specified in the capability to request service while a relay-specific capability is spent at a 
specific Tor relay to build a \sys-enhanced circuit through the relay. 
A Tor client can use both capability types simultaneously by visiting a \sys-enhanced 
service provide through a \sys-enhanced Tor circuit. In Figure~\ref{fig:sys_arch}, 
we intentionally separate our two capability use cases for clear presentation. 
\section{The Access Authorities}\label{sec:AAs}
\sys relies on a set of fully distributed and partially trusted access authorities (AAs)
to manage network capabilities. We assume AAs to be honest-but-curious,
meaning that they follow protocol, but seek to derive additional 
information about Tor clients.  An AA can be deployed by the Tor Project,
service providers (\eg large CDNs like Cloudflare), or third parties. 
Each AA is a conceptually centralized entity. However, an AA can \emph{distribute} 
its operations among multiple servers to achieve high availability.

\subsection{Capability Seeds} 
AAs expect valid \emph{capability seeds} from Tor clients to issue \emph{pre-capabilities}, 
which are the basis for deriving spendable capabilities. For flexibility, we 
intentionally keep the definition of capability seeds broad: any resource 
that is readily available to Tor users, but costly to scale, can be adopted  
as capability seeds. Reasonable choices include proof-of-work schemes 
(\eg solutions to CAPTCHAs or computational puzzles) and anonymous 
monetary resources.  \emph{\sys does not assume that 
	capability seeds can distinguish bots from humans.}  
Rather, botnets can still obtain  more capability seeds 
than legitimate Tor users. Instead, \sys employs capability
seeds as a form of anonymous identities that enable both service providers
and the Tor network to control access by each Tor client.

In this paper, we elaborate on two types of capability seeds (\ie solutions to CAPTCHAs and computational puzzles) and 
further discuss how \sys can incorporate more types of seeds in \S~\ref{sec:extending_AAs}. 
One key challenge of using anonymous capability seeds is to ensure that clients do not have  
to solve endless challenges while browsing the web and meanwhile ensure their activities 
are unlinkable. \sys proposes a capability renewal protocol to address this 
challenge (\S~\ref{sec:site_pre_cap_computation}). 

Although CAPTCHAs can be deployed using publicly available 
libraries like Google's reCAPTCHA~\cite{reCAPTCHA}, \sys needs additional components 
to support computational puzzles. At a very high level, \sys's puzzle 
system design is similar to Portcullis~\cite{portcullis}. However, \sys's puzzle system does make 
a great improvement over Portcullis: it can explicitly 
bound the percentage of CPU cycles that any client can spend on 
solving puzzles. As a result, the puzzle system  
can bring all bots down to the percentage that normal users prefer to 
use for puzzle computation, which significantly reduce the computation disparity 
between the normal clients and bots. 
For better readability, we defer detailed design for \sys's puzzle 
system in \S~\ref{sec:appendix:puzzles}.

\subsection{Per-Seed Rate Limiting}\label{sec:rate_limiting}
Each AA accepts only one type of capability seed. The rate at which a seed can request 
pre-capabilities is limited. In particular, an AA publishes two rate limiters: one     
determines the maximum rate at which a capability seed can request pre-capabilities used 
for accessing \sys-enhanced service providers and the other one 
determines the maximum rate at which a seed can request pre-capabilities  
used for \sys-enhanced circuit creation. 
Based on these per-seed rate limiters published by all AAs, 
both service providers and Tor can 
configure a set of rules to fulfill their access policies. 
This paper presents two concrete examples. In \S~\ref{sec:site_cap_spending}, 
we elaborate on a design that enables a site to bound an adversary's 
achievable service request rate through Tor using self-defined parameters. 
In \S~\ref{sec:design_relay_specific_cap}, we present a design 
that allows Tor to prevent botnets from creating 
numerous Tor circuits to conduct various abuses. To improve 
readability, detailed settings of these rate limiters 
will be discussed when presenting these access policies.

\subsection{Key Management}
Each AA maintains two pairs of keys for signing pre-capabilities, 
and each of them is dedicated for one capability type. Each AA must publish the public 
key of both key pairs, for instance, via the Tor network consensus, 
to ensure other entities (\eg Tor clients, relays and service providers) 
can verify the AA's signatures. An AA can periodically 
renew its keys, but at any time only two key pairs from the AA are valid. 
After receiving signed pre-capabilities from an AA, Tor clients must verify that 
the AA uses proper keys before using the pre-capabilities for accessing 
service providers or Tor. This prevents a malicious AA from using more keys 
simultaneously to partition the anonymity set.  Finally, each AA is associated 
with a long-term fingerprint to uniquely identity the AA, similar to the fingerprint 
of a Tor relay.  

\subsection{Extending the Access Authorities}
\label{sec:extending_AAs}
Besides Tor, content delivery networks (\eg Cloudflare or Akamai) also have direct incentives to
deploy and control their own set of access authorities to mitigate Tor-emitted 
abuses while serving anonymous connections.  In fact, Cloudflare is working on 
an independent implementation of a system whose design goals are similar to our AAs~\cite{cloudflare_bybass}, 
although they focus on addressing the usability issues for Tor users when visiting Cloudflare-powered websites.  

Finally, semi-trusted third parties such as social network operators (\eg Google, Facebook, 
and Twitter), may also run access authorities (shown as \emph{TTP AA} in
Figure~\ref{fig:sys_arch}) based on pre-agreed terms.  To prevent account information 
leakage to Tor and service providers, Tor users only authenticate themselves to 
the social network operators. Service providers or Tor only learn a single bit of
information: whether a Tor client has a valid account (\ie capability seed) or not. 
\section{\sys-Enhanced Site Access}\label{sec:serviceCapability}
We now elaborate on the capability design for accessing 
\sys-enhanced service providers such as websites.
To mitigate the tension between service providers and Tor users,
our key observation is that service providers should not treat all connections from one Tor 
exit relay equally since each exit relay is shared by many Tor users.
Instead, accountability should be enforced at the granularity of Tor clients so that each 
service provider can throttle malicious Tor clients without blocking legitimate Tor users.
To this end, \sys designs site-specific capabilities that
allows a service provider to enforce self-selected
access rules on anonymous Tor connections.

\subsection{Pre-capability Design}
\label{sec:site_pre_cap_computation}
Before visiting a \sys-enhanced site, a Tor client must first request pre-capabilities from an AA. 
The client is free to choose any AA based on what capability seed the client prefers to give. 
To request a pre-capability, the client \first provides a valid capability seed to its selected AA 
and \second provides blinded information for the AA to compute pre-capabilities. 
The client can hide its  network identity from the AA, for instance, by using Tor.

\parab{Capability Seed Validation.} 
Depending on the accepted type of capability seed, an AA performs corresponding seed verification. 
For instance, if an AA accepts proof-of-work schemes, it needs to verify that solutions to the presented 
challenge are correct. Further, an AA needs to ensure that the pre-capability request rates by any capability 
seed does not exceed the two rate limiters discussed in \S~\ref{sec:rate_limiting}. 
Since each AA maintains separate rate limiters and signing keys for two pre-capability types 
(\ie either for \sys-enhanced service access or for \sys-enhanced Tor circuit creation), 
Tor clients must specify the pre-capability type in their requests (in this section, it is for 
accessing service providers). In \S~\ref{sec:site_cap_spending}, we will explain how a 
site defines its access policies based on these pre-capability release rate limiters published by all AAs.

\parab{Information Required to Compute Pre-capabilities.}
To request pre-capabilities, the client provides its selected AA the following set of information 
$\{ \mathbb{S},  n, \mathcal{T}_s, \mathcal{F}\}$, where $\mathbb{S}$ is the domain name 
of site that the client is going to visit, $n$ is a 128-bit cryptographic nonce generated by the client,  
$\mathcal{T}_s$ is a universally agreed timestamp to indicate the freshness 
of the information and $\mathcal{F}$ is fingerprint of the selected AA. All information is 
blinded~\cite{blind} by the client to avoid information leakage to the selected AA. 

The set of information is designed to prevent abuse. In particular, $\mathbb{S}$ is used to make the capability 
site-specific to prevent capability double-spending at different sites.  The nonce $n$ is added to ensure 
the uniqueness of each pre-capability, which in turn ensures the uniqueness of each capability. 
The $\mathcal{T}_s$ indicates the freshness of pre-capabilities so that expired ones are nullified 
automatically. The client is required to use Tor's daily generated fresh random number~\cite{Tor_random_number} 
as $\mathcal{T}_s$ such that at any time all valid capabilities have the exact same timestamp. This design eliminates 
the possibility of information leakage cased by timestamp abuse. $\mathcal{F}$ is added to allow other entities 
(\ie clients, Tor relays and service providers) to use correct public keys to verify signatures. 

\parab{Computation.}  
Upon validation of the client's pre-capability request, the AA computes pre-capabilities using the 
blinded information provided by the client. Pre-capabilities computed by an AA $\mathcal{A}_i$ 
are denoted by $\mathcal{P}_{\mathcal{A}_i}$. Then we have 

\begin{equation}\label{eqn:P_N}
\mathcal{P}_{\mathcal{A}_i} = \{\mathbb{S} ~|~ n ~|~ \mathcal{T}_s 
~|~ \mathcal{F}_{\mathcal{A}_i} \}^b ~|~ \mathcal{S}^b_{\mathcal{A}_i},
\end{equation}
where $\mathcal{F}_{\mathcal{A}_i}$ is $\mathcal{A}_i$'s fingerprint, 
$\mathcal{S}^b_{\mathcal{A}_i}$ is ${\mathcal{A}_i}$'s 
blind signature over the set of blinded information 
$\{\mathbb{S} ~|~ n ~|~  \mathcal{T}_s ~|~ \mathcal{F}_{\mathcal{A}_i} \}^b$, 
and $|$ represents concatenation throughout the paper. 

\parab{Pre-capability Renewal.} 
One key challenge for designing pre-capabilities is to ensure 
that Tor clients do not have to repeatedly solve challenges when browsing the web. 
A strawman design is that an AA can issue many (\ie a few hundred) pre-capabilities 
for each solved challenge. However, this strawman design has at least two 
shortcomings: \first it breaks the site-specific pre-capability design since 
the client may not be able to forecast the sites that it is going to visit so as  
to provide these blinded information immediately after solving challenges; 
\second the design makes it easier for automated bots to accumulate  
pre-capabilities, weakening the entire system. 

To combat these problems, we propose a pre-capability renewal protocol. In particular, 
when a client first presents its challenge solution (\ie capability seed) to an AA, the AA issues 
the client a unforgeable \emph{pseudonym} $\mathcal{I} = \{r ~|~ \phi\}$ where $r$ is a 
random 128-bit nonce and $\phi$ is the AA's signature over $r$. Later on, the client 
presents $\mathcal{I}$ as a proof of validation when requesting new pre-capabilities 
from the AA, allowing the client to bypass future challenges. Not only does the site-specific pre-capability 
design hold with this design, but also the AA can account each 
pre-capability request on a specific solved challenge (\ie capability seed) to enforce 
the per-seed rate limiting described in \S~\ref{sec:rate_limiting}. Each pseudonym has a validation 
period determined by the AA. Clients with expired pseudonyms are required to solve new 
challenges to obtain new pseudonyms that are unlinkable to previous ones. 

\emph{Impact of the Pseudonym on Anonymity.} 
Different from the prior pseudonym-based anonymous blacklisting systems~\cite{chaum, chen}, in which 
a user interacts with a service provider using a persistent pseudonym, the pseudonym in our pre-capability 
renewal protocol is transient and never presented to both service providers 
and Tor relays. The pseudonym in our protocol is only linked with a specific challenge solution 
served as an anonymous capability seed. Since a Tor client presents its pseudonym to an  
AA through Tor, the AA cannot link the pseudonym with the client. 
Further, since all site-related information sent to the AA is blinded, 
the pseudonym is unlinkable with any site access as well. 
Thus, using pseudonym in our protocol does not impact Tor users' anonymity.

\subsection{Site-specific Capability Design}\label{sec:siteCapability}
After obtaining $\mathcal{P}_{\mathcal{A}_i}$, the Tor client \emph{unblinds} the signature  
using its secret blind factor to produce the unblinded version of the pre-capability, which 
is the capability spendable at a specific site. In particular, 
\begin{equation}\label{equ:site_capability}
\mathcal{C} = \mathbb{S} ~|~ n ~|~  \mathcal{T}_s ~|~  \mathcal{F}_{\mathcal{A}_i} ~|~ \mathcal{S}_{\mathcal{A}_i}
\end{equation}

The capability $\mathcal{C}$ contains a set of unblinded information that allows 
the site $\mathbb{S}$ to perform capability verification when the client presents $\mathcal{C}$
to access the site, as detailed in \S~\ref{sec:site_cap_spending}.

Employing blind signature is the key to ensure that \sys preserves Tor's privacy guarantee. 
First, signatures from the AAs prevent unauthorized entities from issuing capabilities. 
Second, using blind signature avoids disclosing any site-related information to the AAs since the blinded 
information sent to the AAs is unlinkable with the ``plain'' information produced by the client. 
Such unlinkability further ensures the unlinkability between the client and its capability spending 
even if the AAs could collude with the site, which preserves online anonymity of Tor users. 
We provide a formal security proof in \S~\ref{sec:security_analysis}.

\subsection{Site-Specific Capability Spending}\label{sec:site_cap_spending}
\parab{Capability Validation.}
Tor clients spend site-specific capabilities at \sys-enhanced sites to request services. 
Upon receiving capabilities, a \sys-enhanced site first validates them before subsequent 
processing. A site-specific capability is valid if \first it encloses an authentic 
signature from an AA; \second it encloses a domain name that is consistent with the site; 
\third the capability is not expired (\ie $\mathcal{T}_s$ is the fresh random number 
released by Tor); and \four the capability is not nullified by the site. If any of these 
conditions does not hold, the site rejects this capability to deny access. If a  
CDN provider (\eg Cloudflare) processes capabilities on behalf of its powered sites, the 
second rule is passed as long as the enclosed domain is owned by one of the 
CDN provider's customers. In the fourth rule, whether a capability is nullified or not is 
decided by the site's access policies, as detailed below. 

\parab{Site-Defined Access Policies.}
Once a site-specific capability is validated, the site accepts the Tor client's service request. 
Since the major form of Tor abuse is that automated bots use Tor to conduct various malicious 
activities against the site~\cite{cloudflare-trouble} (\eg content scraping, vulnerability scanning, comment  
spamming and so forth), the site needs to further control the number of service requests 
(\eg HTTP requests) allowed by each capability. We clarify that each site can have 
its own definition of service requests. Once a Tor client's service request count exceeds a 
threshold, the site nullifies the current capability and requires a new site-specific capability for 
subsequent service requests.  Recall that the pre-capability request rate by each 
client is limited by the AAs through the per-seed rate limiting design in \S~\ref{sec:rate_limiting}. 
Thus, together with these rate limiters, it is possible for the site to design access policies 
so as to bound a strategic adversary's  service request rate using self-selected 
parameters, as detailed below.

\textbf{\emph{Policy Definition.}}
Assume the following set of access authorities $\{\mathcal{A}_0, \mathcal{A}_1,..., \mathcal{A}_n\}$ 
are deployed, and each authority accepts one type of capability seed. In this context, 
the site defines its access policy as $\{w_0, w_1, ..., w_n\}$ where $w_i$ is the number 
of service requests allowed by one valid site-specific capability issued by the access authority $\mathcal{A}_i$.

We now formulate $\{w_0, w_1, ..., w_n\}$ mathematically. 
We denote the set of capability seeds by $\{s_0, s_1, ..., s_n\}$ and authority $\mathcal{A}_j$ 
accepts seed $s_j$. Let $c_j$ denote the cost of obtaining a capability seed $s_j$. 
We denote the cost of obtaining one network identity (\ie IP address) by $\lambda$. 
Let $r_j$ denote the maximum rate at which a seed $s_j$ can request pre-capabilities 
(for accessing service providers) from authority $\mathcal{A}_j$. 
Assume that for any client connecting to the site directly 
without using Tor, the site allows a maximum 
service request rate $\widetilde{\mathcal{O}}$ 
before either blocking the client or forcing the client to 
solve challenges. Then to bound a strategic adversary's  
service request rate by using Tor, the site derives $\{w_0, w_1, ..., w_n\}$ 
to ensure that the following condition is satisfied for any set of 
parameters $[\alpha_0, \alpha_1, ..., \alpha_n]$ where 
$\alpha_i \in [0,1]$ and $\sum_{i=0}^n \alpha_i = 1$. 

\begin{equation}\label{eqn:site_access_rule}
\sum_{i=0}^n ~\frac{\alpha_i \cdot \lambda }{c_i}\cdot r_i \cdot w_i 
\leq \epsilon \cdot \widetilde{\mathcal{O}}, 
\end{equation}
where $\epsilon$ is a site-defined parameter.

\textbf{\emph{Policy Correctness.}} 
The parameters $[\alpha_0, \alpha_1, ..., \alpha_n]$ represent the adversary's strategy of  
purchasing various types of capability seeds. Thus, if formula~(\ref{eqn:site_access_rule}) holds for any strategy,
the site can guarantee that the maximum Tor-emitted service request rate achieved by an adversary 
when spending $\lambda$ on purchasing capability seeds is no greater than $\epsilon \cdot \widetilde{\mathcal{O}}$. 
Thus, if an adversary that spends a certain  amount of resources on obtaining network identities 
can access the site with rate $\mathcal{O}$ without using Tor, then the maximum rate that the adversary 
can request service from the site by using Tor is no greater than $\epsilon \cdot \mathcal{O}$, 
given that the adversary spends the same amount of resources on acquiring capability seeds. 
Equivalently, in order to achieve the same service request rate, the adversary has to 
spend $1/\epsilon$ times as many resources when launching attacks through Tor 
as it spends when launching attacks natively without using Tor. To ensure that 
formula~(\ref{eqn:site_access_rule}) holds for any attacker strategy, we choose 

\begin{equation}\label{eqn:site_rules_solved}
w_i \leq \epsilon \cdot \frac{c_i \cdot \widetilde{\mathcal{O}}}
{\lambda \cdot r_i}, ~\forall i \in [0,n]
\end{equation}

\textbf{\emph{Policy Enforcement.}}
If $w_i = 1$, then each capability is usable for exactly one service request. The site  
can enforce this by suppressing service requests with duplicate capabilities, for example, 
through the use of a Bloom filter. If $w_i > 1$, then statistically more than 
one service request should be allowed for each capability. To enforce this, 
the site stops accepting a capability with probability $1/w_i$, and then adds the 
capability to the duplicate suppressor. However, multiple service requests carrying the same 
capability can trivially be linked by the site. We discuss how to address this issue through 
system parameterization below. Finally, if $w_i < 1$, then each capability is 
accepted with probability $w_i$, and exactly one service request is allowed for each 
accepted capability. 

\textbf{\emph{Policy Parameterization.}} 
We now discuss the parameterization of $w_i$. First, to compute $w_i$, the site 
does not need to exactly know $c_i$. Instead, the site simply needs to assign specific 
weights to these capability seeds based on its policies. 
Further, with an ideal parameterization, $w_i$ should be exactly one since \first 
no capability is spendable on more than one service
request to ensure unlinkability and \second no additional capabilities are required for 
a single service request to avoid extra computation and networking 
overhead. However, it is difficult to reach the ideal parameterization since $r_i$ 
is chosen by the authority $\mathcal{A}_i$ that is unaware of the site's configurations 
$\epsilon$ and $\widetilde{\mathcal{O}}$. In addition, configurations can vary greatly 
among different sites so that an ideal parameterization for one site could be undesirable 
for others. 

To address the problem, \sys sets $r_i$ such that (with high probability) a Tor client 
can obtain enough capabilities so that it is feasible  for the client to present a unique capability for 
each TCP connection to the site. This ensures that the client can achieve the highest level of unlinkability 
offered by Tor, \ie service providers only see  TCP connections from Tor exit relays. 
We clarify that it is the client who determines how to spend its capabilities across TCP 
connections (as described below). The above parameterization is adopted only to ensure 
that spending a unique capability for each TCP connection is a feasible strategy for the client. 
A reasonable setting of $r_i$ can be estimated based on the live Tor measurement in~\cite{PrivCount}, 
which finds that during a 10-minute interval, each Tor client opens about $24$ web streams.   
In practice, the authority $\mathcal{A}_i$ should enforce $r_i$ over a longer period of time 
(\eg few hours) to accommodate usage burst. 

Note that when an AA $\mathcal{A}_k$ is deployed by the site itself, system 
parameterization for $\mathcal{A}_k$ is easier since 
the site determines the rate limiters for issuing pre-capabilities. 

\parab{Capability Spending by Tor Clients.}
Given $r_i$, some sites may end up with rules $w_i > 1$, \ie one capability 
is allowed for multiple service requests. In this case, the site needs to send a response to 
indicate whether a capability is nullified or not. Tor clients are free to determine their capability 
spending strategies. For instance, a Tor client can send $w_i$ service requests using the 
same capability within a single TCP connection (due to HTTP keep-alive), which still ensures 
the highest level of unlinkability. Or the client may choose to spend one capability 
across multiple TCP connections to allow trans-TCP linkability. We note that if a Tor client uses the 
default setting of Tor Browser, it already allows trans-TCP linkability since the Tor Browser 
uses session cookies. For a site that has $w_i$ less than $1$, it can enforce such policies 
by accepting one capability with probability $w_i$ and for each accepted capability, 
the site allows only one service request. 
\section{\sys-Enhanced Tor Access} \label{sec:TorCapability}
In this section, we detail the capability design for accessing the \sys-enhanced Tor network. 
The current Tor network suffers from a variety of botnet abuses such as large scale C\&C abuse~\cite{botnet_abuse_1,
	botnet_abuse_2,botnet_abuse_3, hopper2013protecting, hidden_service_abuse},  
relay flooding attacks~\cite{cellflood,sdos} and traffic analysis~\cite{mittal:ccs11,murdoch:sp05}. 
These abuses lead to various bad results, including poor system performance for legitimate Tor users, 
de-anonymization threats and bad reputation for Tor. The root cause of these attacks is that botnets 
can create an arbitrary number of Tor circuits without any limitation. Enforcing local rate limiting for 
circuit creation at each relay is unlikely to stop these attacks since a strategic botnet can 
instruct each bot to enumerate all relays to circumvent the local rate limiting. 

With \sys, Tor can globally control circuit creations by any client using  
our capability scheme. In particular, when \sys is activated, clients are required to 
possess valid capabilities in order to create \sys-enhanced circuits (to be incrementally 
deployable, circuit creation requests without valid capabilities are de-prioritized in case 
of congestion). Then, by controlling the rate at which a Tor client can obtain  
capabilities, \sys can explicitly limit the client's circuit creation rate.  

\subsection{Relay-Specific Capability Design}\label{sec:design_relay_specific_cap}
To create a three-hop \sys-enhanced circuit, a Tor client $\mathbb{U}$ needs to obtain 
three capabilities, each of them being specific to a relay on the circuit. The design of 
relay-specific capabilities is identical to that of site-specific capabilities, except for the following.  
\first During pre-capability requesting, the client needs to specify 
the proper pre-capability type, \ie it is for Tor-enhanced circuit creations. 
Further, to request a pre-capability specific to a relay $\mathbb{R}$, 
the client encloses the fingerprint of relay $\mathbb{R}$ (rather than any site domain)
in the set of blinded information sent to its selected AA. \second Relay-specific capabilities are 
spendable at \sys-enhanced relays (not at any sites) for creating Tor-enhanced circuits 
through the relays. The relays first validate received capabilities (based on a set of rules 
similar to those defined in \S~\ref{sec:site_cap_spending}) before extending circuits.

We clarify that to request pre-capabilities, Tor clients do not have to use 
\sys-enhanced circuits to reach the AAs. Thus, there is no deadlock for bootstrapping \sys. 
Another alternative is pre-installing few relay-specific capabilities 
on Tor clients so that using \sys-enhanced circuits to bootstrap the system is viable.    

\parab{Policy Definition.} Similar to site-specific capabilities, relay-specific 
capabilities enable Tor to enforce access rules for its relays. In this paper, 
we propose to use capabilities to control the circuit creation rate by any Tor client   
so as to mitigate those aforementioned botnet abuses against Tor. In particular, assume 
the following set of AAs  $\{\mathcal{A}_0, \mathcal{A}_1,..., \mathcal{A}_n\}$ are deployed 
and authority $\mathcal{A}_i$ accepts a type of capability seed $s_i$. In this context, 
Tor defines its access policies as $\{q_0, q_1, ..., q_n\}$, where $q_i$ is the maximum rate at which 
a capability seed $s_i$ can request pre-capabilities (for creating Tor-enhanced circuits) from 
authority $\mathcal{A}_i$. Then in order to bound a Tor client's circuit creation rate, 
$\{q_0, q_1, ..., q_n\}$ should satisfy the following condition for any attacker 
strategy $[\alpha_0, \alpha_1, ..., \alpha_n]$ where $\alpha_i \in [0,1]$ and 
$\sum_{i=0}^n \alpha_i = 1$. 

\begin{equation}\label{eqn:Tor_access_rule}
\sum_{i=0}^n ~ \frac{\alpha_i \cdot \lambda \cdot q_i}{3 \cdot c_i} 
\leq \mathcal{T}, 
\end{equation}
where $\lambda$ is the cost of getting one network identity, 
$c_i$ is the cost for obtaining one capability seed $s_i$ and 
$\mathcal{T}$ is the maximum circuit creation rate allowed for a client, 
which is a parameter controlled by Tor. We note that the constant $3$ 
appears in above formula since a standard  Tor circuit contains $3$ 
relays and each of them consumes a relay-specific capability.

To ensure the correctness of formula~(\ref{eqn:Tor_access_rule}) 
for any attacker strategy, we choose  
\begin{equation}\label{eqn:Tor_access_rule_solved}
q_i \leq \frac{3 \cdot c_i \cdot \mathcal{T}}{\lambda}, ~\forall i \in [0,n]
\end{equation}

\parab{Parameterization.}
Similar to how sites determine its access rules using Equation (\ref{eqn:site_rules_solved}), 
to compute $q_i$, the Tor Project needs to assign certain weights to these capability seeds. 
Further, a proper configuration of $\mathcal{T}$ can be determined based on the live Tor 
measurements in \cite{PrivCount}. In particular, during an 10-minute interval, PrivCount~\cite{PrivCount} 
estimates that a Tor client opens about $4$ Tor circuits. Thus, the maximum rate $\mathcal{T}$ at 
which one Tor client can create circuits should be close to $4$ per ten minutes. 
In practice, each AA should enforce $q_i$ over a longer period of time (\eg few hours) 
to accommodate usage bursts and relay churn.

\subsection{Capability Exchange for Tor OSes}\label{sec:capability_trans}
The design of relay-specific capabilities needs to be augmented with a \emph{capability exchange 
protocol} to better support Tor onion services (OSes). In particular, a Tor onion server (itself 
runs a Tor client) needs to open many Tor circuits in order to serve all its clients (referred to as OS-clients). 
Although a Tor hidden server can continue to use legacy Tor circuits to serve its OS-clients, we do design a 
capability exchange protocol to enable onion servers to use \sys-enhanced circuits as well. 

The design intuition is that a OS-client requests a new type of capability, \ie \emph{trans-capability}, from the AAs, 
and sends it to the OS, which subsequently redeems the trans-capability at the AAs for new 
pre-capabilities. The trans-capability, accounted on the capability seed of the OS-client, 
anonymously informs the AAs that the hidden server needs to create a new \sys-enhanced 
circuit to serve the OS-client. For better readability, the detailed design of the protocol is deferred 
in \S~\ref{sec:appendix:trans_capability}. 

\section{Security Analysis}\label{sec:security_analysis}
In this section, we perform a formal security analysis for the impact of \sys on Tor users' anonymity.  
Let $\mathbb{N}_T$ denote the set of Tor clients that request pre-capabilities 
from the AAs, and subsequently present capabilities to access service providers or Tor relays. 
We first present two useful lemmas on unlinkability.

\subsection{Lemmas}\label{sec:lemmas}\begin{lemma}\label{lemma:1}
Consider any client $\mathbb{U}\in\mathbb{N}_T$. By colluding with each other, both the AAs 
and a service provider $\mathbb{W}$ gain only negligible advantage over random guessing  
when trying to link a specific Tor-emitted site access with the client $\mathbb{U}$.
\end{lemma}

\begin{proof}
We first specify the notations used in the proof. Let $\mathbb{V}$ denote a Tor-emitted site 
access to $\mathbb{W}$ initiated by the Tor client $\mathbb{U}$. Note that the definition of a 
site access is decided by $\mathbb{W}$. Let $\mathcal{C}$ denote the service-specific capability 
that $\mathbb{U}$ sends to $\mathbb{W}$ to support the site access $\mathbb{V}$. 
Let $\mathcal{P}$ denote the pre-capability used by $\mathbb{U}$  to compute $\mathcal{C}$.

Since the client $\mathbb{U}$ can use Tor to connect to the AAs when requesting the 
pre-capability $\mathcal{P}$, in the ideal case, $\mathbb{U}$ is unlinkable with $\mathcal{P}$. 
However, to ensure that our lemma still holds in the worst case when 
Tor's unlinkability is broken by adversaries, we assume the AA $\tilde{\mathcal{A}}$ that issues 
$\mathcal{P}$ can link $\mathcal{P}$ with the client $\mathbb{U}$. Thus, the service provider $\mathbb{W}$ 
and other AAs can have such linkability as well by colluding with $\tilde{\mathcal{A}}$.

Next, we prove the lemma by contradiction. Assume that the AAs 
and $\mathbb{W}$ can design an algorithm $\mathcal{K}$ that enables the AAs and $\mathbb{W}$ 
to link the site access $\mathbb{V}$ with the client $\mathbb{U}$. 
Since the site access $\mathbb{V}$ is linkable with the capability $\mathcal{C}$ 
(as $\mathcal{C}$ is presented to the site to support the access $\mathbb{V}$)
and the client $\mathbb{U}$ is linkable with the pre-capability $\mathcal{P}$ (based on 
the above worst-case assumption), designing the algorithm $\mathcal{K}$ is 
equivalent to designing another algorithm $\mathcal{K}'$ that enables the AAs 
and $\mathbb{W}$ to link the capability $\mathcal{C}$ with the pre-capability $\mathcal{P}$. 

In \sys's design, $\mathcal{P}$ is the blinded message signed by the AA $\tilde{\mathcal{A}}$ 
(\ie the blind-signer), and $\mathcal{C}$ is the unblinded version of $\mathcal{P}$ produced by 
the client $\mathbb{U}$ using a secret factor unknown to the blind-signer. Thus, the 
problem of designing $\mathcal{K}'$ to link $\mathcal{P}$ with $\mathcal{C}$ is  
the same as designing an algorithm $\mathcal{K}''$ that allows a blind-signer to link the 
blinded message it signs to the unblinded message without knowing the secret factor, which 
is impossible in a blind signature~\cite{blind, partial}. This contradiction proves that the 
hypothetical algorithm $\mathcal{K}$ does not exist, indicating both AAs and $\mathbb{W}$ gain 
only negligible advantages of linking an specific site access $\mathbb{V}$ with client $\mathbb{U}$ 
via collusion. We clarify this lemma does not claim that colluding among multiple entities does 
not pose a risk for Tor; it only proves that \sys does not introduce any further risk even if multiple entities 
collude with each other. 
\end{proof}

Using the similar reduction proof as  
Lemma \ref{lemma:1}, we can prove the following lemma. 

\begin{lemma}\label{lemma:2}
Consider any client $\mathbb{U}\in\mathbb{N}_T$. 
By colluding with each other, both the AAs and Tor 
relays gain only negligible advantage over random guessing when trying to link a 
specific relay access (\ie \sys-enhanced circuit creation) with client $\mathbb{U}$.
\end{lemma}

\subsection{Information Leakage Analysis}
Given the above two lemmas, we now analyze the impact of \sys on Tor user anonymity. 
We measure the possible information leakage to an arbitrary service provider $\mathbb{W}$ 
based on  \emph{degree of anonymity}~\cite{entropy1,anonymity_degree}. Our analysis uses 
information-theoretic entropy~\cite{entropy2} as the measure of information contained 
in a probability distribution. Recall that $\mathbb{N}_T$ denote the set of \sys-upgraded Tor clients. 
Given an arbitrary capability-enhanced site access (\ie an access supported by a valid capability), 
$\mathbb{W}$ believes that with probability $p_i$, the access originates from client $i$ in $\mathbb{N}_T$. 
Thus, $\mathbb{W}$ maintains a probability distribution $I$ for all anonymous accesses. Then, the entropy 
(\ie the information contained in the distribution $I$) is defined as 
$H_{\mathbb{W}} = -\sum_{i \in \mathbb{N}_T} p_i \cdot \log_2(p_i).$

Based on the unlinkability proven in Lemma \ref{lemma:1}, we have $p_i = \frac{1}{N_T}$, 
where $N_T$ is the size of the anonymity set $\mathbb{N}_T$. Thus, the entropy after introducing \sys is 
\begin{equation}
H_{\mathbb{W}} = \log_2 N_T. 
\end{equation}

Next, we analyze the system entropy before introducing \sys. 
Let $\mathbb{N}$ denote the entire set of anonymous Tor clients. 
Notice that the current Tor network protects $\mathbb{W}$ from linking 
a (native) site access with a specific Tor client. Thus, given the 
entire anonymous client set $\mathbb{N}$, the maximum entropy $H_M$ of 
the system k is $H_M= \log_2 N$, where $N$ is the size of the anonymity set $\mathbb{N}$.  

Thus, based on the definition in \cite{entropy1, anonymity_degree}, the degree of anonymity $d$ after 
introducing \sys is 
\begin{equation}\label{equ:anonymity_degree}
d = 1 - \frac{H_M-H_{\mathbb{W}}}{H_M} = \frac{\log_2 N_T}{\log_2 N}
\end{equation}

\parab{Anonymous Set Analysis.} Given Equation~(\ref{equ:anonymity_degree}), 
the information leakage is determined by the size of the anonymous client set before and 
after \sys is introduced. Therefore, once all Tor clients are upgraded to support 
\sys, there is no information leakage at all. Thus, eventually, \sys completely 
preserves the privacy guarantee offered by the Tor network. To 
mitigate the one-time privacy issue during the early deployment phase of \sys, 
the Tor Project can require mandatory client upgrades from a certain 
time point to ``force'' all active clients to serve as \sys initiators. 

\section{Implementation}\label{sec:implementation}

In this section, we present the full implementation of \sys. 

\subsection{Capability Implementation}\label{sec:implementation:capability}
We implement capability-related computation  using \textsf{C}, \textsf{Python} and \textsf{JavaScript} to 
consider various usage scenarios. For instance, the capability design can be directly built into the Tor software written 
in \textsf{C} (as shown in \S~\ref{sec:implementation_tor}), or it can be implemented as a plugin for the Tor browser, 
which executes capability-related computation in \textsf{JavaScript} (as shown in \S~\ref{sec:implementation:AAs}). Websites  
may compute capabilities using any language. Thus, we use \textsf{Python} as an example due to its popularity in web applications.

We use the RSA algorithm to perform capability-related cryptographic operations such as blind signing. 
The $\textsf{C}$ implementation uses the OpenSSL library~\cite{openssl} and the \textsf{Python} implementation 
imports the PyCrypto module~\cite{pycrypto}. Since no standardized $\textsf{JavaScript}$ 
library for computing blind signatures is available, we develop our own library based on $\textsf{crypto-js}$~\cite{cryptojs} 
and $\textsf{BigInt}$~\cite{bigint}, two libraries that allow us to perform computation (\eg modulo) for 
very large prime numbers in $\textsf{JavaScript}$. We benchmark the capability computation overhead 
in \S~\ref{sec:evaluation_capability_overhead}.

\subsection{Implementation of the Access Authority}\label{sec:implementation:AAs}
For an AA accepts CAPTCHAs as capability seeds (referred to as CAA), we implement it as a 
web server that deploys Google's reCAPTCHA~\cite{reCAPTCHA} service. 
For an AA accepts computational puzzles (referred to as PAA), 
it accepts puzzle solutions over HTTP (or HTTPs) requests. These AA servers define 
a customized HTTP header (\textsf{X-Capability}) to carry \sys-related cryptographic tokens such as 
pseudonyms and pre-capabilities. To make the implementation transparent to clients (\ie no 
client-end network stack modifications are required), the AA servers add \textsf{X-Capability} in the 
\textsf{Access-Control-Allow-Headers} HTTP header option. 

Although Tor clients can access PAA using native HTTP libraries, the CAA needs to be accessed 
using browsers. Thus, we implement a Firefox add-on (referred to as \textsf{CapJS}) 
to execute \sys-related cryptographic operations in browsers. In real-world deployment, 
the add-on should be developed by trusted entities (\eg the Tor project) 
and signed by Mozilla so that Tor users can install it on their Tor browsers. 

\begin{figure}[t]
\centering
\mbox{
\subfigure{\includegraphics[width=0.98\linewidth]{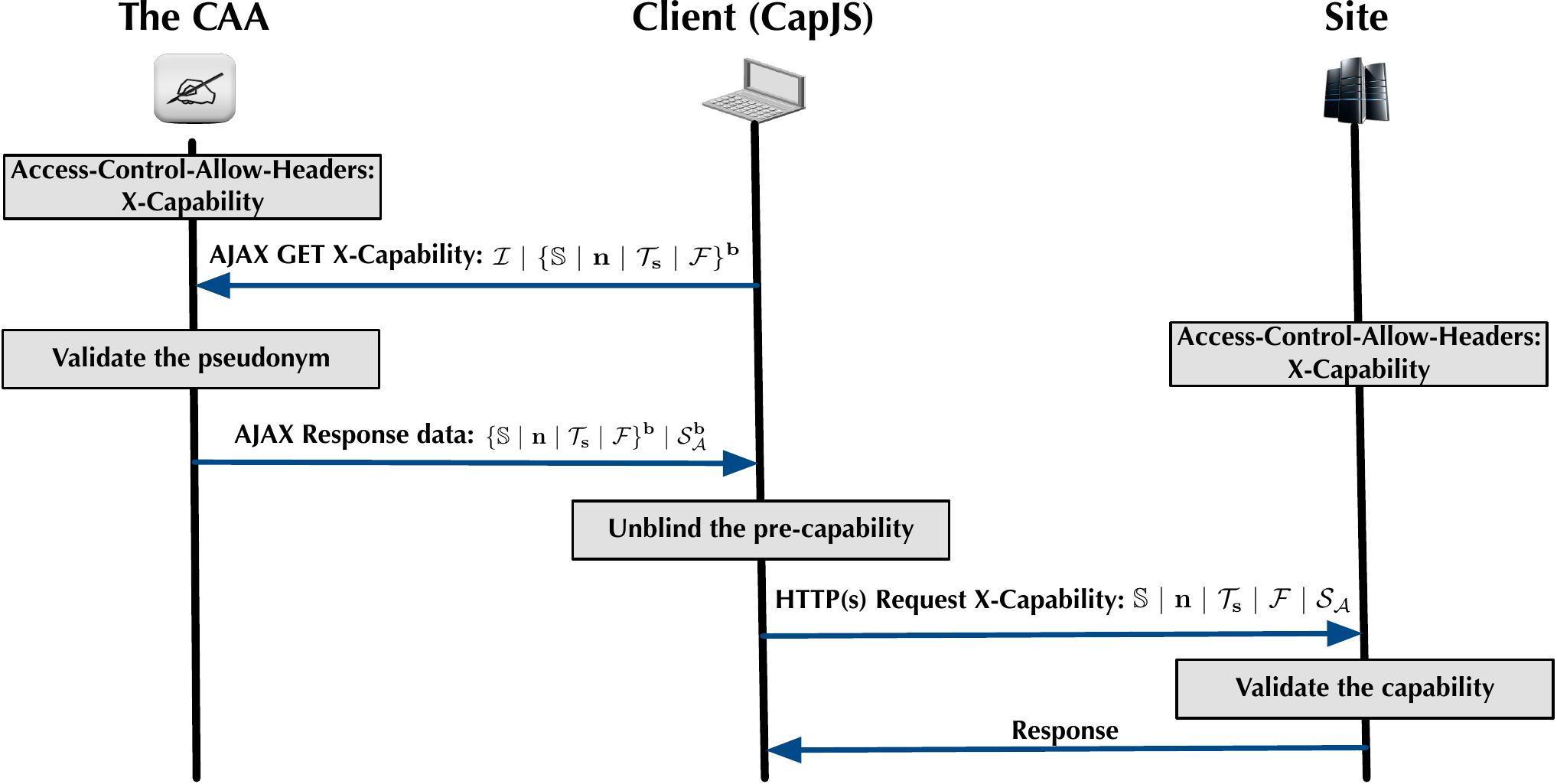}}
}
\caption{\sys-enhanced site access by a Tor client with \textsf{CapJS} installed on its browser.}
\label{fig:site_cap_implementation}
\end{figure}

\parab{\textsf{CapJS} Design.}
When a Tor client connects to a CAA server, \textsf{CapJS} checks cookies to determine whether a pseudonym $\mathcal{I}$ 
issued by the same CAA server is locally cached. If so, \textsf{CapJS} then puts  $\{\mathcal{I} ~|~ \{ \mathbb{S} ~|~ n ~|~ t_s~\mathcal{F}\}^b\}$
into the \textsf{X-Capability} header, where $\{ \mathbb{S} ~|~ n ~|~ t_s~\mathcal{F}\}^b$ is the set of blinded information described in the 
pre-capability design (\S~\ref{sec:site_pre_cap_computation}). If no pseudonym issued by the same CAA is available, \textsf{CapJS} only  
puts the set of blinded information into the \textsf{X-Capability} header.  With this customized HTTP header, \textsf{CapJS}  
sends an AJAX GET to the CAA server. 

After receiving the AJAX request, the CAA server inspects the \textsf{X-Capability} header. If a valid pseudonym is 
retrieved, the CAA server computes a pre-capability for the client using the blinded information carried in the header. 
Otherwise, the CAA server loads a reCAPTCHA challenge page for the client. Once the challenge is successfully 
solved, the CAA server computes a pre-capability by signing the blinded information, as well as a pseudonym 
for the client. These tokens are returned to the client in a JSON object responding to the client's AJAX GET request. 

After receiving a response from the CAA server, \textsf{CapJS} inspects the received data object to 
retrieve the pre-capability and the pseudonym (if applies). The pre-capability 
is then unblinded to produce a capability, and the pseudonym is cached for future use.

\subsection{\sys-enhanced Site Access}\label{sec:implemention_site}
To serve \sys-enhanced Tor clients, the  deployment required at service providers is lightweight. In particular, a site simply 
needs to add \textsf{X-Capability} in its \textsf{Access-Control-Allow-Headers} HTTP header option to allow \textsf{CapJS} 
to pass site-specific capabilities in the header. \textsf{CapJS} is responsible for sending capabilities to corresponding 
sites if the client visits multiple \sys-enhanced service providers, and potentially enforcing the capability 
spending policies discussed in \S~\ref{sec:site_cap_spending}, Upon receiving capability-enhanced service request, 
the site verifies the received capability using the rules discussed in \S~\ref{sec:site_cap_spending} to enforce its 
desired access policies.  Figure \ref{fig:site_cap_implementation} depicts the workflow of a site access by a client 
with \textsf{CapJS} installed on its browser.

\subsection{\sys-enhanced Tor Circuit}\label{sec:implementation_tor}
We now discuss the implementation of \sys-enhanced Tor circuit creation. 

\begin{figure}[t]
	\centering
	\mbox{
		\subfigure{\includegraphics[scale=0.4]{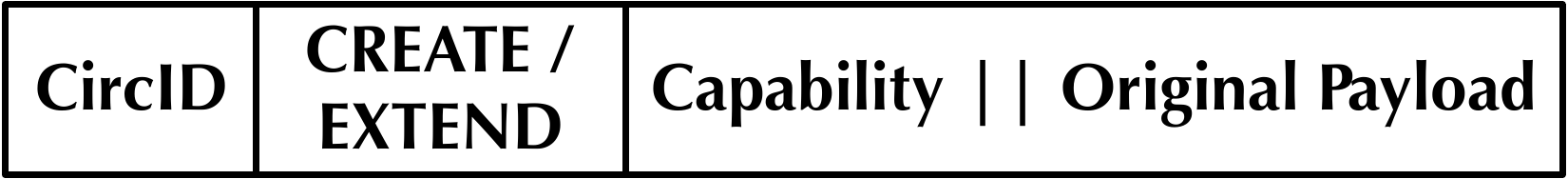}}
	}
	\caption{\sys-enhanced Tor circuit creation.}
	\label{fig:tor_cell}
\end{figure}

\parab{Tor Source Code Modification.}
We modify the Tor software source code to integrate our capability design into Tor circuit creation. The native Tor circuit 
creation proceeds as follows. The onion proxy (OP) on a Tor client first sends a \textsf{CREATE} cell 
containing the first half of the Diffie-Hellman handshake to a guard relay, which responds with a 
\textsf{CREATED} cell containing the second half of the handshake. To extend the circuit to a new relay $\mathbb{R}_e$, 
the OP sends a \textsf{RELAY\_EXTEND} cell (specifying the address of $\mathbb{R}_e$ and a new secret) to the last relay 
$\mathbb{R}_m$ on the partially-created circuit. $\mathbb{R}_m$ copies the received information into a new \textsf{CREATE} 
cell, and forwards the cell to $\mathbb{R}_e$. 

To create a capability-enhanced circuit, in addition to these original cells, the OP further sends a valid relay-specific capability 
to each hop. In our prototype, the OP \emph{prepends} a capability to the payload of the \textsf{CREATE} cell when connecting to  
the guard relay. Capabilities for subsequent relays are prepended to corresponding \textsf{RELAY\_EXTEND} cells. 
Figure~\ref{fig:tor_cell} illustrates the modified cell structure. Each relay first verifies the received capability based on 
the rules defined in \S~\ref{sec:design_relay_specific_cap} before processing the onionskin carried in the remaining 
payload. Since capability verification is much cheaper than the onionskin processing, this design saves the relay considerable 
compute resources for processing  bogus circuit creations without valid capabilities. Alternatively, a relay-specific 
capability can be carried via a customized cell. In this case, the cell should be sent together with the 
\textsf{CREATE} cell (or \textsf{RELAY\_EXTEND} cell) to avoid additional RTTs.

To validate our implementation, we  test the modified Tor source code in Shadow~\cite{shadow},  
a safe development environment to run real Tor source code in a private Tor network. Via log analysis, our test 
experiments show that our implementation properly embeds relay-specific  capabilities into the workflow of Tor circuit creation.

\parab{Live Tor Interaction.} 
Since live Tor relays do not run our modified Tor source code, we cannot create \sys-enhanced circuits directly 
through live Tor relays. Thus, we implement \emph{another} prototype to interact with live Tor relays during the capability-enhanced 
circuit creation. We defer implementation details in \S~\ref{sec:appendix:live_tor}.  

\section{EVALUATION}\label{sec:evaluation}
Our evaluation centers around the following. 

\parab{\sys Introduces Small System Overhead.} 
We show that capability-related operations introduce small overhead compared with the typical Tor circuit creation latency 
(\S~\ref{sec:evaluation_capability_overhead}). Further, we show that the deployment overhead of 
the AAs is small. For instance, the AAs collectively need only $11$ cores to support the entire 
set of current Tor users (\S~\ref{sec:eva_AAs}).  

\parab{\sys Effectively Enforces Site-Defined Policies.}
We demonstrate that \sys enables a site to  effectively enforce its access policies on anonymous 
Tor connections, \ie the site can bound service request rate by any strategic adversary via self-defined 
parameters (\S~\ref{sec:eva_site_policies}).

\parab{\sys Mitigates Various Abuses Against Tor.} 
Based on real data collected by Tor, we demonstrate \sys can mitigate large scale C\&C abuse and prevent 
cell flooding attacks against the Tor network (\S~\ref{sec:eva_mitigate_Tor_abuse}).

\subsection{Capability Computation Overhead}\label{sec:evaluation_capability_overhead}
\begin{table}[t]
\caption{The computational time (in microseconds) for capability-related cryptographic operations.}
\label{tab:computation_overhead}
\centering
\resizebox{0.98\linewidth}{!}{%
\begin{tabular}{l r r r l}
\toprule

\textbf{Operation} &
\textbf{Mean} &
\textbf{Median} &
\textbf{Std. Dev.} &
\textbf{Language} \\

\midrule

\multirow{3}{*}{Generation} & 232.0 & 232.0 & 0.1 & \textsf{C} \\
& 253.7 & 253.6 & 0.3 & \textsf{Python} \\
& 27,320.0 & 27,240.0 & 245.5 & \textsf{JavaScript} \\

\midrule

\multirow{3}{*}{Verification} & 25.6 & 25.6 & 0.0 & \textsf{C} \\
& 32.0 & 32.0 & 0.1 & \textsf{Python} \\
& 355.5 & 354.3 & 5.3 & \textsf{JavaScript} \\

\midrule

\multirow{3}{*}{Blinding} & 3.5 & 3.5 & 0.0 & \textsf{C} \\
& 46.3 & 46.3 & 0.1 & \textsf{Python} \\
& 18.1 & 18.1 & 0.3 & \textsf{JavaScript} \\

\midrule

\multirow{3}{*}{Unblinding} & 2.4 & 2.4 & 0.0 & \textsf{C} \\
& 7.0 & 7.0 & 0.0 & \textsf{Python} \\
& 64.8 & 64.7 & 6.8 & \textsf{JavaScript} \\

\bottomrule
\end{tabular}}
\end{table}

In this section, we benchmark the overhead of capability-related computation in \textsf{C}, \textsf{Python} and \textsf{JavaScript} 
on our testbed. All results are obtained using a single 3.30GHz Intel i3-3120 core. We perform $10,000$ runs to learn the mean, median,
and standard deviation of the computation times for a single capability generation, verification,  information blinding and unblinding. 
We perform experiments for various RSA key lengths. Results shown in Table~\ref{tab:computation_overhead} are obtained when the RSA 
key length is $1024$. The overall computational overhead is small. For instance, it takes an AA ${\sim}230$ microseconds  
in \textsf{C} to compute a pre-capability. And  a single capability verification takes ${\sim}25$ microseconds in \textsf{C}. 
A blinding and an unblinding operation by Tor clients can be finished in ${\sim}3$ and ${\sim}2$ microseconds, respectively, in \textsf{C}. 
The implementations in \textsf{C} and \textsf{Python} have comparable performance since  PyCrypto internally wraps \textsf{C} code.
Although it is more expensive to perform signing and verifying in \textsf{JavaScript}, the overhead of blinding and unblinding operations 
(performed by Tor clients) in \textsf{JavaScript} is comparable with other languages. The AAs, relays and service providers 
can adopt more efficient languages such as \textsf{C} and \textsf{Python} to perform capability generation and verification.

\begin{figure}[t]
\centering
\mbox{
\subfigure[The \# of cores required.\label{fig:precapability_release:a}]{\includegraphics[scale=0.23]{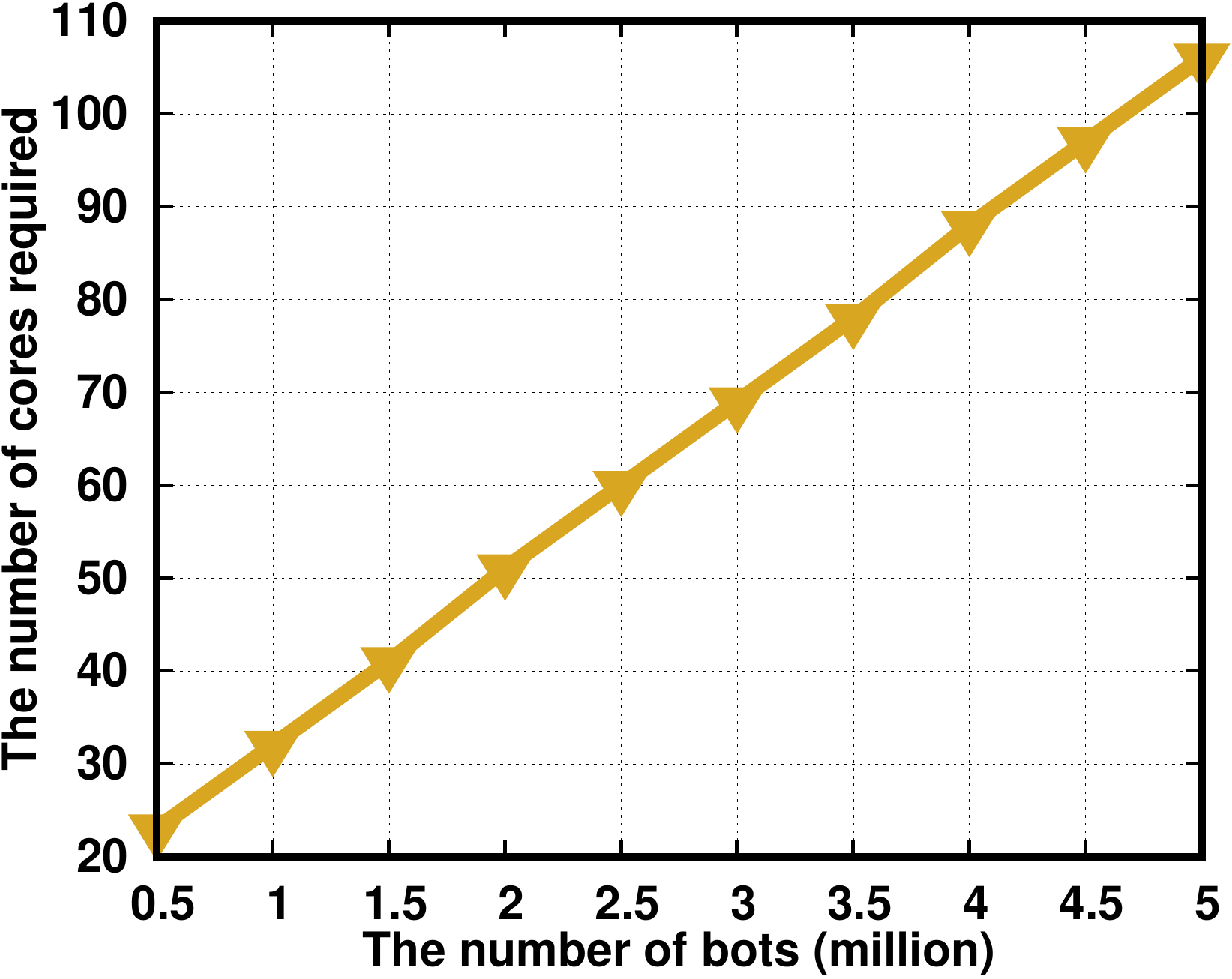}}
\subfigure[Pre-capability issuing latency.\label{fig:precapability_release:b}]{\includegraphics[scale=0.23]{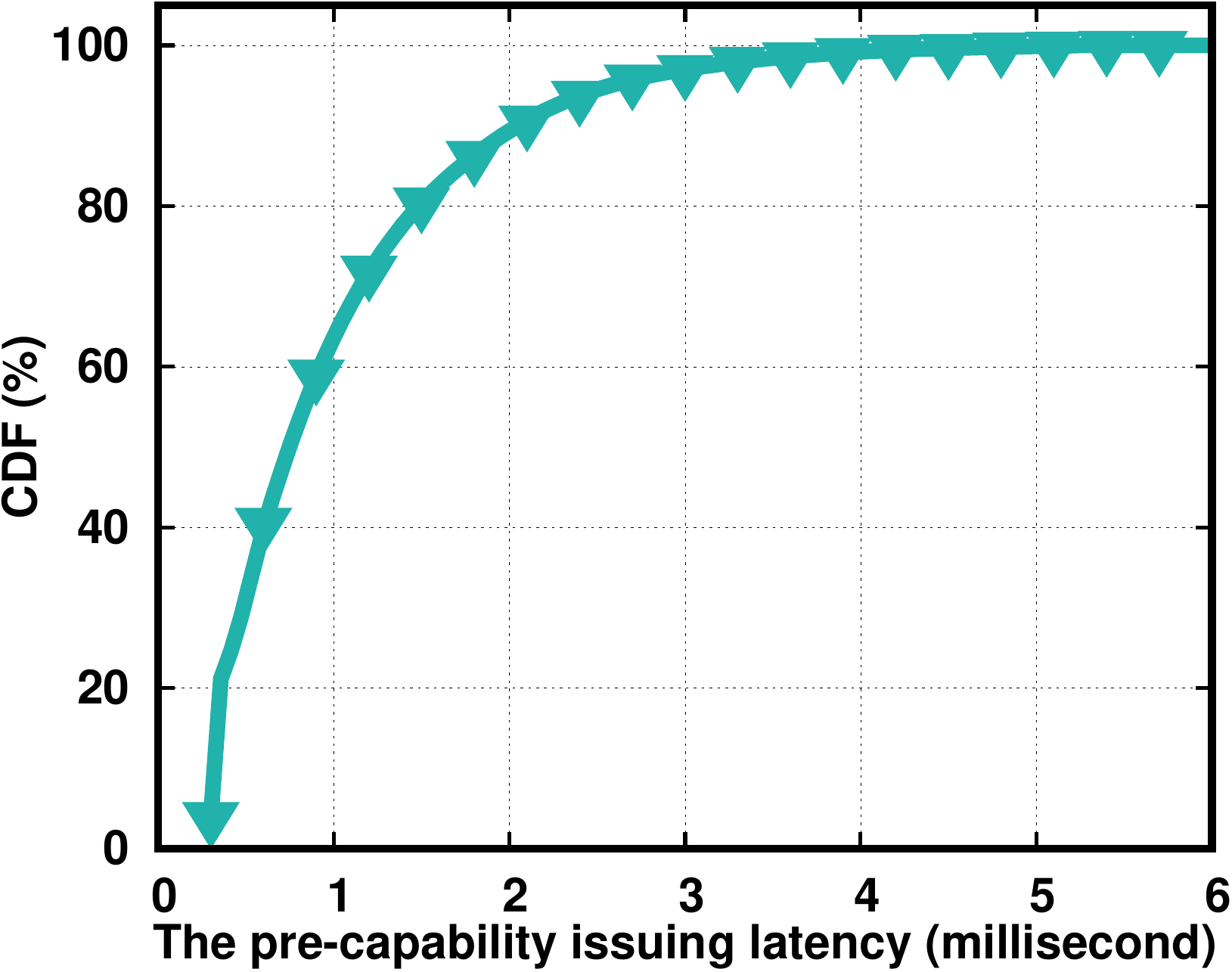}}
}
\caption{Pre-capability release analysis. Figure~\ref{fig:precapability_release:a} shows the number of cores required to prevent the AAs from 
being overwhelmed by numerous requests. Figure~\ref{fig:precapability_release:b} plots pre-capability release latency benchmarked on 
our testbed.}\label{fig:precapability_release}
\end{figure}

\subsection{Deployment Overhead of the AAs}
\label{sec:eva_AAs}
In this section, we evaluate the  deployment overhead of the AAs. We first estimate the compute resources needed by the AAs 
to support the pre-capability computation for all Tor users. Then we evaluate the pre-capability issuing latency using 
the AAs deployed on our testbed. 

\parab{Collective Compute Resources Needed.} 
To estimate compute resources required from the AAs, we need to estimate the amount of 
pre-capability requests from all Tor clients. Recall that each AA server maintains two rate limiters 
for issuing pre-capabilities: $r$ for issuing site-specific pre-capabilities and $q$ for issuing 
relay-specific pre-capabilities (\S~\ref{sec:rate_limiting}). We estimate $r$ and $q$ using the live 
Tor measurement results in~\cite{PrivCount}. In particular, during  a 10-minute interval, PrivCount~\cite{PrivCount} 
estimates that each Tor client opens about 24 TCP streams and $4$ circuits. Further, PrivCount~\cite{PrivCount} 
counts about $710,000$ unique clients during a 10-minute interval. Combing these, we estimate the collective 
pre-capability request rate from all Tor clients is about $44,000$ per second. Since it takes one core $0.23$ milliseconds to 
issue one pre-capability, the AAs collectively need about $11$ cores to support the entire set of current Tor users. 

In practice, the AAs should be over-provisioned to prevent an adversary from overwhelming them  via 
massive pre-capability requests. We clarify that such flooding attack aims to exhaust the AAs'  
compute resources rather than their network bandwidth  (bandwidth-oriented volumetric DDoS attacks can be 
prevented by hosting the AAs on well-provisioned cloud~\cite{MiddlePolice}). Figure~\ref{fig:precapability_release:a} 
plots the number of cores required in order to withstand different-sized botnets. The results 
show that the AAs need about $100$ cores to withstand a 5-million node botnet.

\parab{Pre-Capability Release Latency.} 
We now evaluate pre-capability release latency using the AAs deployed on our testbed (\S~\ref{sec:implementation:AAs}). 
We define the pre-capability release latency as the time required for an AA server to process a pre-capability request, excluding  
networking latency and other user-introduced latency (\eg the time required for solving challenges). 
We provision eight servers on our physical testbed as AA servers in this experiment. We double-threaded each AA server so that 
the eight AA servers collectively have 16  cores. To emulate pre-capability requests from the entire set of Tor users, 
we develop a requester that generate requests at the rate of $44,000$ per second. To send each request, the requester 
randomly picks one of the 8 AA servers. The results, plotted in Figure \ref{fig:precapability_release:b}, 
show that the pre-capability release latency is less than few milliseconds, which is over 2 orders of magnitude smaller 
than the typical circuit creation time (0.7s based on our live Tor measurements in \S~\ref{sec:appendix:live_tor}).

\subsection{Enforcing Site-Defined Policies}\label{sec:eva_site_policies}
In this section, we demonstrate that \sys enables a site to enforce site-defined access policies for  
anonymous Tor connections. As a result, a site can explicitly bound 
a strategic adversary's service request rate using self-defined parameters.

\parab{Access Policies.}
For evaluation purpose, we assume that the site assigns equal weights to both types of 
capability seeds, \ie $c_0$ (for CAPTCHA solutions) and $c_1$ (for puzzle solutions) 
in Equation~(\ref{eqn:site_rules_solved}) are the same. However, the actual costs, 
denoted by $c_0'$ and $c_1'$, can be different from $c_0$ and $c_1$. 
Further, base on the measurements in  \cite{botnet_cost, CAPTCHA_IN_ECO}, 
we assume $c_0'$ is close to $\lambda$ (the cost for obtaining one network identity). 

We evaluate three strategies that a site may use to define its access policies. The first 
strategy (referred to as \emph{basic strategy}) is that the site accepts all Tor-emitted requests 
with valid capabilities. In the second strategy (referred to as \emph{rate limiting strategy}), 
the site enforces a maximum service request rate $r_{\max}$ for \emph{all} Tor-emitted requests 
with valid capabilities. In the third strategy, besides rate limiting, the site further performs weighted 
fair queuing (WFQ) to serve requests: rather than serving all valid Tor-emitted requests in one FIFO queue, 
requests with capabilities obtained using CAPTCHA solutions and puzzle solutions are served in two separate 
FIFO queues weighted equally. The third strategy (referred to as \emph{WFQ strategy}) prevents 
one type of seed from overwhelming the other one. 

\begin{figure}[t]
	\centering
	\mbox{
		\subfigure{\includegraphics[scale=0.4]{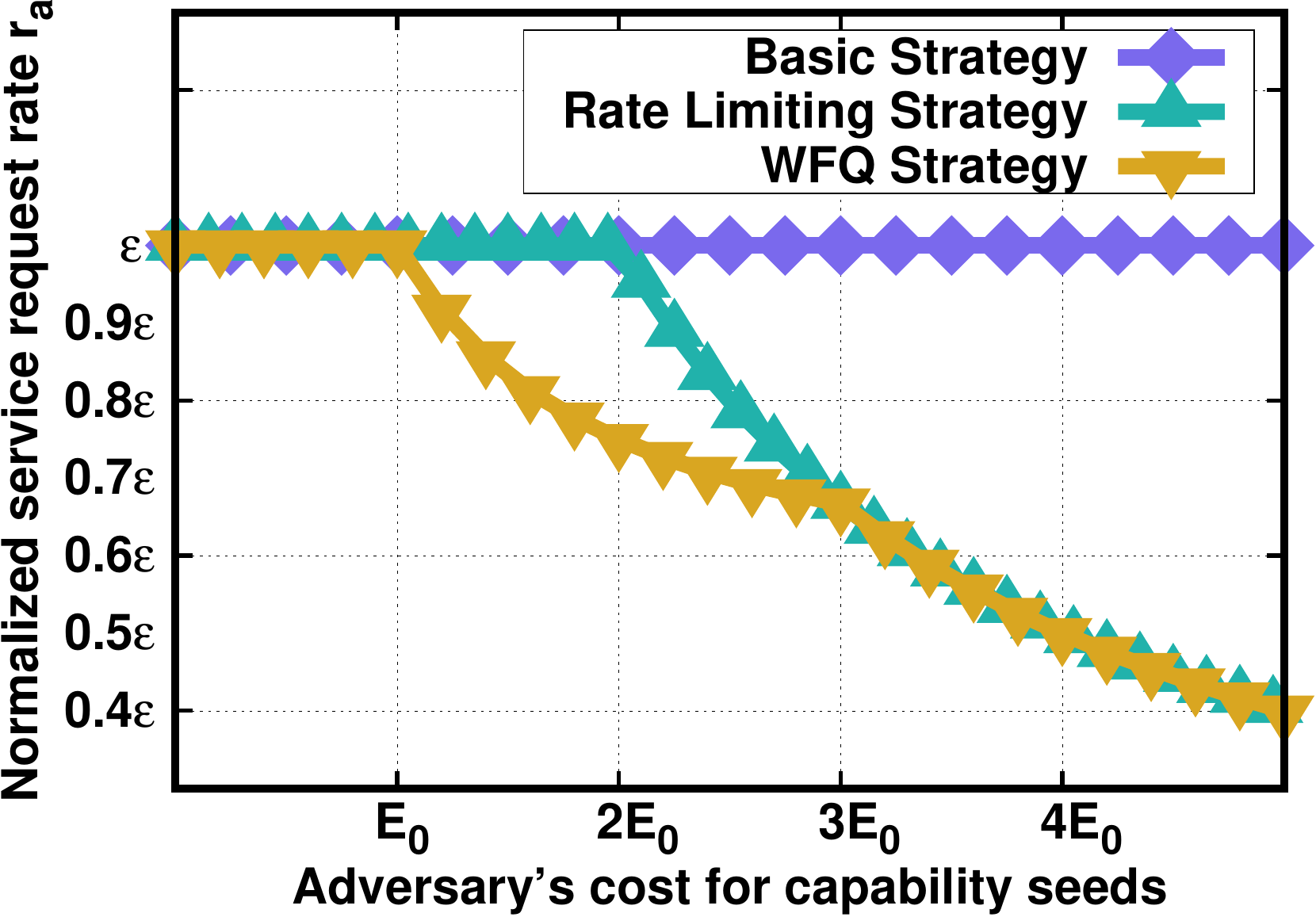}}
	}
	\caption{\sys enables a site to bound an adversary's service request rate using 
		self-defined parameter $\epsilon$. $E_0$ is the adversary's cost at the first point of diminishing returns  
		when the site uses the WFQ strategy.}	\label{fig:policy_enforcement}
\end{figure}

\parab{Policy Enforcement.} We now study an adversary's service request rate through Tor when it invests a certain 
amount of money on acquiring capability seeds. Define  $ k = {c_0'}/{ c_1'}$. 
We first present the evaluation results for $k=0.5$ in Figure~\ref{fig:policy_enforcement} and then 
extend our discussion to arbitrary $k$. For any amount of investment, the adversary's service request 
rate through Tor (denoted by $r_a$) is normalized to the service request rate obtained 
when the adversary connects to the site directly without using Tor.

Since $c_0' {<} c_1'$ given $k{=}0.5$, the adversary's optimal strategy is spending all investment on solving CAPTCHAs. 
Thus, we have $r_a {=} \epsilon$, where $\epsilon$ is the site-configurable parameter defined in Equation~(\ref{eqn:site_access_rule}). 
When the site adopts the basic strategy, $r_a$ remains the same as the adversary increases its investment. 
However, for the other two strategies, $r_a$ will reach a point of \emph{diminishing returns} as the adversary's 
investment further increases (as shown in Figure~\ref{fig:policy_enforcement}). In particular, when the site adopts 
the rate limiting strategy, the point of diminishing returns is reached when the collective service request 
rate from the adversary and all legitimate Tor clients exceeds $r_{\max}$. In Figure~\ref{fig:policy_enforcement}, 
we denote the adversary's cost at this point by $2E_0$. After that, further increasing investment actually reduces $r_a$ 
since no more Tor-emitted requests are allowed by the site. 

When the WFQ strategy is adopted, $r_a$ experiences two points of diminishing returns as the adversary's cost increases, 
as shown in Figure~\ref{fig:policy_enforcement}. The first one happens when the collective service request rate from all Tor 
clients using the optimal seed (CAPTCHAs in this evaluation) exceeds $\frac{r_{\max}}{2}$. After this point, the adversary has  
to use sub-optimal seeds in order to further get services. As a result, $r_a$ starts to decline from the optimal rate  
$\epsilon$. The second point of diminishing returns is reached when the collective Tor-emitted service request rate exceeds $r_{\max}$.    

\parab{General Results.} Our further analysis (deferred in \S~\ref{sec:appendix:bound}) proves that for any 
$k$, $r_a \leq \epsilon$ if $k\leq1$ and $r_a \leq k \cdot \epsilon$ if $k\ge1$. Thus, regardless of the actual cost of 
obtaining capability seeds, the adversary's service request rate is bounded by $\Theta(\epsilon)$. 
This result holds no matter which strategy the site adopts and how many types of capability seeds are 
accepted by \sys.

\subsection{Mitigating Botnet Abuse Against Tor}\label{sec:eva_mitigate_Tor_abuse}
In this section, we perform Tor-scale evaluations to demonstrate the following. 

\parab{Mitigating Botnet C\&C Abuse in Tor.} 
Based on the real data collected from the large scale botnet C\&C abuse against Tor happened during Aug-Sep 2013, we show that 
Tor clients suffered from very high circuit  failure rates (${\sim}$40\%) during the abuse. Then we demonstrate that \sys 
effectively mitigates the abuse by reducing failure rates by ${\sim}$74\% (\S~\ref{sec:evaluation:CC}). 

\parab{Mitigating Tor-targeted DDoS Attacks.} 
We demonstrate that \sys significantly increases Tor's resilience against cell flooding attacks that aim 
to paralyze the Tor network via excessive circuit creation requests (\S~\ref{sec:evaluation:tor_ddos}).

\emph{Tor-scale Simulator.} We aim to show that \sys is able to mitigate the harm
that a multi-million botnet can do to Tor.  While we do have a \sys
implementation that runs on Shadow~\cite{shadow} (see \S~\ref{sec:implementation_tor}), 
we would run into scalability issues with simulating millions of Tor clients.  Further, Shadow is unable to help us
simulate the cryptographic overhead that botnets would impose on Tor relays~\cite{shadow_QA}.  Due to these 
shortcomings, we developed our own simulator.  We faithfully implement Tor's path selection
algorithm~\cite{path_spec} and validate the correctness of our implementation
by comparing relays' selection probability with the ones published by Tor
Atlas~\cite{tor_atlas}.  We sampled the computational capacity of relays from
Barbera \ea's work that was based on live Tor measurements~\cite{cellflood}.

\begin{figure}[t]
	\centering
	\mbox{
		\subfigure{\includegraphics[scale=0.4]{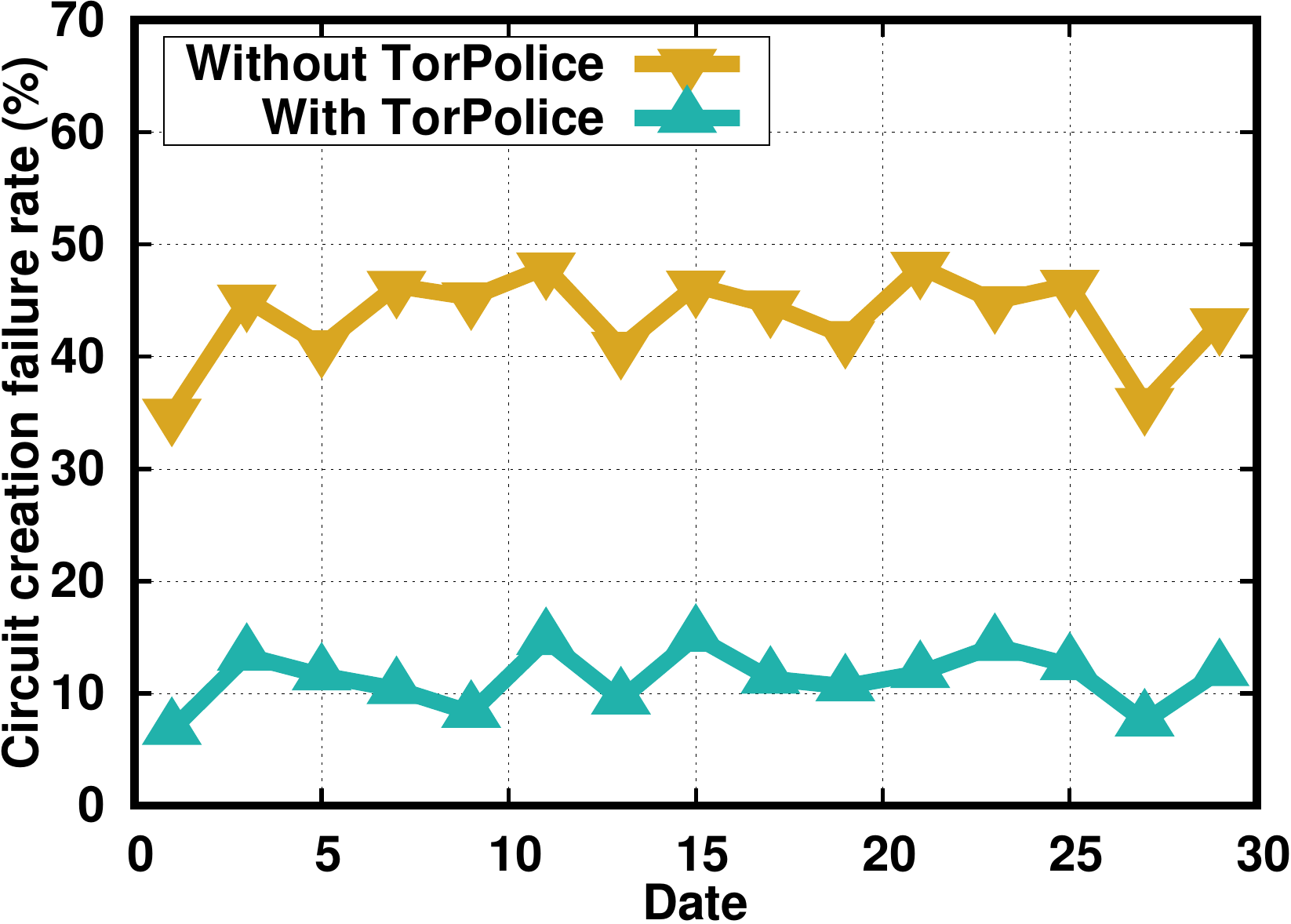}}}
	\caption{Circuit creation failure rates when Tor faced a multi-million node botnet C\&C abuse. 
		\sys can reduce the average failure rate by ${\sim}74$\%. }	\label{fig:circuit_failure_rate}
\end{figure}

\subsubsection{Mitigating Botnet C\&C Abuse}\label{sec:evaluation:CC}
~\newline
In this section, we study the botnet C\&C abuse that happened during Aug-Sep 2013, when Tor's 
daily estimated users rapidly increased from 1 million to 6 million. We first show that Tor clients experienced very high 
circuit creation failure rates when Tor was  under this abuse. Then we show \sys effectively mitigates such abuse.

\parab{Circuit Creation Failure Rate.} 
We use the data collected by Tor to estimate the amounts of circuit creations initiated  
by the botnet during the C\&C abuse. To improve readability, we defer the detailed 
mathematical modeling to \S~\ref{sec:appendix:modeling}. Due to the massive circuit creations 
by the botnet, compute  resources of many relays are exhausted, resulting in very high circuit 
creation failure rates, as depicted in Figure~\ref{fig:circuit_failure_rate}. Such high failure rates are 
caused by the following vicious cycle. When the abuse starts, Tor relays begin to drop requests due 
to the lack of compute resources. These initial failures force the bot clients to continuously send requests 
until their circuits are successfully created, which further increases the network load. The resulting 
consequences are that the botnet still managed to use Tor as its primary C\&C channel after numerous trials 
whereas Tor is less usable for  legitimate Tor users since it could require tens of trials before a circuit is 
finally created, resulting a high user-perceived latency. 

\parab{Mitigating Botnet Abuse in Tor.} 
The root cause of such high circuit creation failure rates is that bot clients deviate from 
typical Tor usage pattern, \ie they initiate numerous circuit creation requests without any limitation. As described in 
\S~\ref{sec:design_relay_specific_cap}, \sys allows Tor to explicitly control the circuit creation by any Tor clients. 
Thus, to counter this  abuse, Tor sets its access policies $q_i$ such that the maximum rate at which a client can create 
circuit is $4$ per ten minutes (aligned with live Tor measurements in \cite{PrivCount}). We plot the resulting circuit 
creation failure rates after enforcing the access policies in Figure~\ref{fig:circuit_failure_rate}. Clearly, \sys effectively 
eases the network load and reduces the average failure rate from ${\sim}41\%$ to ${\sim}10\%$, a ${\sim}74$\% reduction. 

In response to the C\&C abuse, the Tor project released a new version (0.2.4.17-rc) that prioritizes the processing of 
onionskins using the \texttt{ntor}~\cite{ntor} protocol since the bot clients used an older version without \texttt{ntor} support. 
Tor's countermeasure reduced the circuit failure ratio to about $20\%$~\cite{hidden_service_abuse}. We clarify that a strategic botnet
could circumvent Tor's defense by changing adaptively (\eg upgrading software). However, \sys offers long-term countermeasures 
that can handle strategic botnets.

\vspace{0.1in}
\subsubsection{Mitigating Tor-targeted DDoS Attacks}
~\newline
As noted in~\cite{hopper2013protecting}, a general concern of attacking a botnet by Tor in case of abuse 
(\eg by blacklisting its hidden servers) is that it may lead to retaliation. For instance, a botnet can easily 
paralyze Tor via excessive circuit creation requests. According to Tor design~\cite{path_spec}, a Tor 
client drops its current guard relay when  circuit failure rate measured by the client is above 
30\%. Via massive circuit creation requests,  an adversary can easily exhaust computation resources 
of the entire set of relays, driving circuit failure rates much higher than this threshold. 
Figure~\ref{fig:Tor_targeted_DDoS} demonstrates  this vulnerability: a moderate-sized 
botnet with hundreds of thousands of bots is enough to cause very high circuit failure ratios. 
When Tor is protected by \sys, however, even a multi-million node botnet can only cause very limited failure rates for 
the current Tor network (represented by the consensus published on May 1st 2017). 

\label{sec:evaluation:tor_ddos}
\begin{figure}[t]
	\centering
	\mbox{
		\subfigure{\includegraphics[scale=0.4]{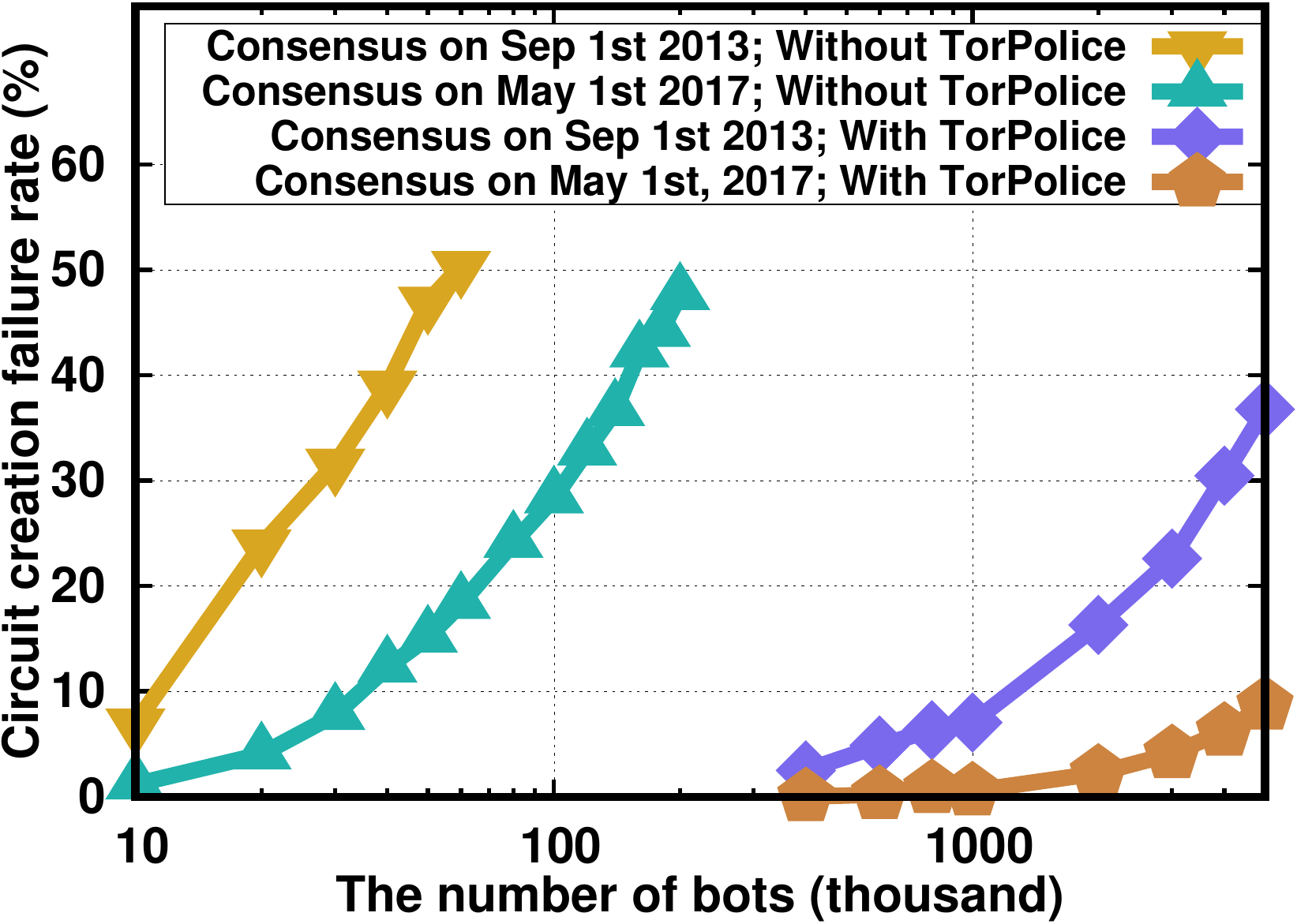}}}
		\caption{Without \sys, a moderate-sized adversary can paralyze Tor via cell flooding attacks. 
		\sys can effectively mitigate this vulnerability.}\label{fig:Tor_targeted_DDoS}
\end{figure}

\balance
\section{RELATED WORK}\label{sec:related}
In this section, we discuss closely related work. 

\noindent\textbf{Capabilities in the Internet.} 
Capability schemes (\cite{SIFF,tva,portcullis,netfence,MiddlePolice}) have been proposed to protect the Internet from 
DDoS attacks. In these approaches,  capabilities specify certain traffic policing rules 
and meanwhile carry cryptographic signatures to ensure correctness. Victims (\eg servers or congested routers) 
police traffic based on received capabilities to stop attacks. Different from \sys, Internet capability designs  
do not consider privacy. Further, some of these capability schemes are difficult to deploy  since they require modification of 
Internet core and client network stack. On the contrary, \sys is readily deployable in Tor with small overhead.

\parab{Anonymous Blacklisting Systems.} 
Anonymous blacklisting systems~\cite{accountable_ano} allow service providers to maintain a ``blacklist'' to explicitly 
block abusive users while serving non-abusive users without breaking anonymity. Anonymous blacklisting systems can be 
categorized into three broad groups: the pseudonym systems~\cite{chaum, chen, damgaard, stubblebine, lysyanskaya}, 
the Nymble-like systems~\cite{AA1, AA2, AA3,AA4,AA5}, and the revocable anonymous credential systems based on  
zero-knowledge proofs~\cite{AA6,AA7,AA9}. These systems either offer pseudonymity instead of full anonymity 
or require a trusted or semi-trusted authority to provide anonymity. \sys is not designed to be a new anonymous blacklisting 
system. Rather, \sys is explicitly designed for Tor, focusing on proposing  a capability-based access control 
framework that allows service providers and Tor to enforce access rules to throttle various botnet abuses while still serving 
legitimate Tor users properly.  Further, \sys's trust is more distributed since its AAs are fully distributed and each of them 
only has a partial view of the entire system. 

\parab{Relay Incentives.}
Tor relay incentive mechanisms~\cite{gold_star, braids, lira, torpath}
are proposed to recruit more relays for the Tor network. Gold Star~\cite{gold_star}, BRAIDS~\cite{braids} 
and LIRA~\cite{lira} incentivize Tor clients to relay anonymous traffic by offering them prioritized Tor services.  
TorPath~\cite{torpath} instead pays relays Bitcoins. By allowing relays to redeem their received 
relay-specific capabilities for various benefits, \sys provides a general framework to support these incentive mechanisms. 
For instance, to support the similar incentive mechanism in \cite{braids}, a relay $\mathbb{R}$ can 
redeem its received generic capabilities to obtain ``prioritized relay-specific capabilities'' from the AAs. 
Then $\mathbb{R}$, as a client, can subsequently spend these prioritized capabilities to create 
premium Tor circuits to get premium services. 
\section{CONCLUSION}\label{sec:conclude}
In this paper, we present \sys, the first privacy-preserving access control framework that allows service providers and 
Tor to enforce self-selectable access policies on anonymous Tor connections so as to throttle various botnet abuses 
while still providing service to legitimate Tor users. \sys leverages blindly signed network capabilities to 
preserve the privacy of Tor users.  We implement a prototype of \sys, and perform extensive 
evaluations to validate \sys's design goals. 

\clearpage
\renewcommand*{\bibfont}{\footnotesize}
\balance
{
	\bibliography{paper}
	\bibliographystyle{acm}
}

\clearpage
\balance
\section{APPENDIX}\label{sec:appendix}

\subsection{Distributed Puzzle Systems}\label{sec:appendix:puzzles}

\begin{figure}[t]
	\centering
	\mbox{
		\subfigure{\includegraphics[width=0.98\linewidth]{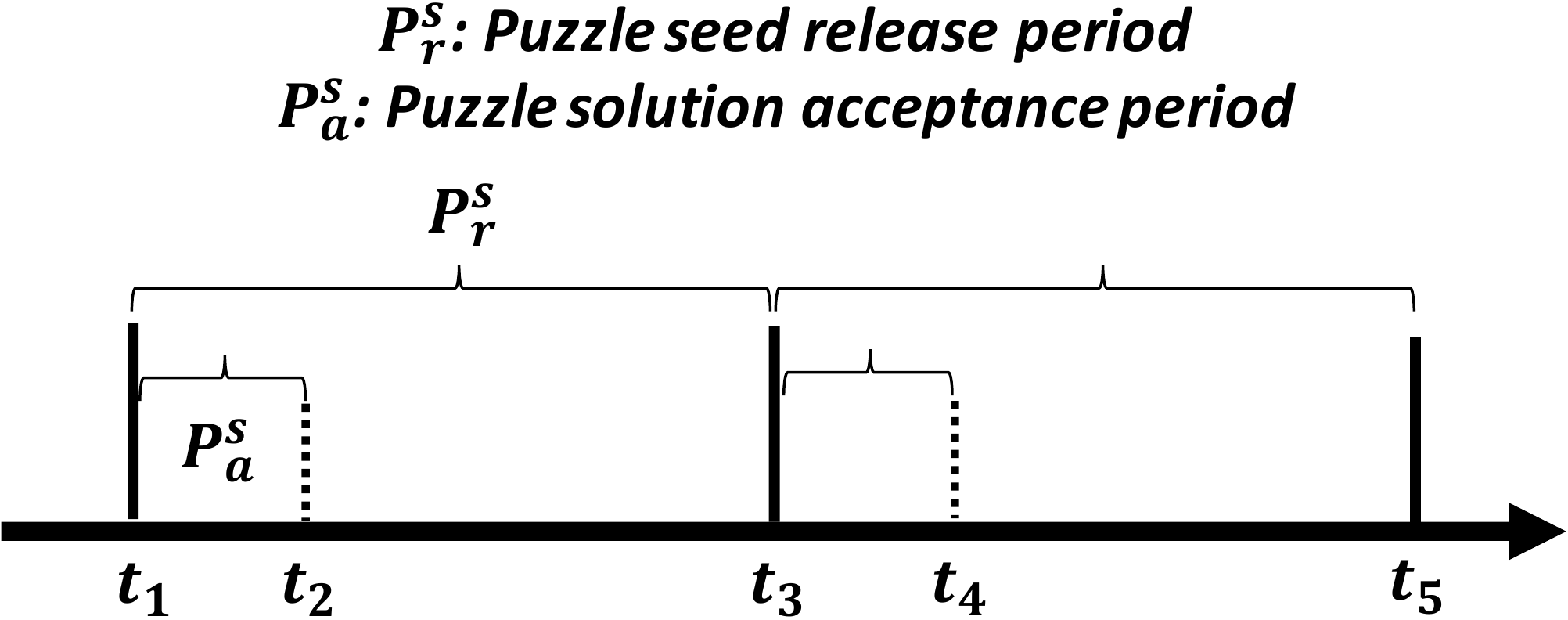}}
	}
	\caption{
		\sys's puzzle system works on the basis of $P^s_r$ and $P^s_a$: 
		One fresh puzzle seed is released in each $P^s_r$ and puzzle solutions 
		are redeemable at AAs for pre-capabilities only within $P^s_a$ in each $P^s_r$.}
	\label{fig:puzzle_period}
\end{figure}

To support puzzle solutions as capability seeds, \sys introduces a distributed puzzle system for distributing computational puzzles. 
Compared with prior systems (\eg Portcullis~\cite{portcullis}), the novelty of \sys's puzzle system is that it can explicitly bound the 
CPU usage by any client for solving puzzles. In particular, legitimate clients do not prefer to use all their CPU cycles to compute puzzles. 
However, automated bots do. To enable access control, prior systems (\eg Portcullis~\cite{portcullis}) would need to 
prioritize requests based on the difficulty level of puzzles since otherwise the bots could overwhelm the system by solving easy puzzles. 
Thus, to compete with automated bots, legitimate clients are forced to use all their CPU cycles to solve puzzles while still at the risk of 
being denied access when competing with automated bots with significant computation resources. On the contrary, by explicitly bounding the 
percentage of CPU cycles allowed for solving puzzles, \sys's can bring all bots down to the percentage that normal clients prefer to 
use for puzzle computation, which significantly reduces the computation disparity between legitimate clients and automated bots. 

\vspace{0.1in}
\subsubsection{Puzzle System Overview}~\\
All computational puzzles are computed based on a series of \emph{puzzle seeds} that 
are released periodically. Tor's existing directory authorities (DAs), for instance, can be used for releasing puzzle seeds. 
The puzzle system works on the basis of two periods, as illustrated in Figure~\ref{fig:puzzle_period}. 
In each \emph{puzzle seed release period} ($P^s_r$), one fresh puzzle seed is released at the 
beginning of the period and no more puzzle seeds will be further released in this period. 
The seed release algorithm (\S~\ref{sec:seed_release}) ensures that the puzzle seeds cannot not be 
pre-computed and each valid seed requires the participance from a majority of all DAs. 

In each $P^s_r$, all puzzles are computed based on the puzzle seed released in the current $P^s_r$. Thus, it is impossible 
to pre-compute solutions for future puzzles even if the puzzle generation algorithm (\S~\ref{sec:puzzle_generation}) is public. 
Similarly, solutions to previous puzzles cannot be used as valid capability seeds in the current $P^s_r$. To bound a client's 
CPU usage for solving puzzles, all puzzle solutions have to be returned to the AAs within the \emph{puzzle solution 
	acceptance period ($P^s_a$)}.  Late solutions will not be accepted. Thus, the percentage of CPU usage for solving puzzles 
is bounded by $P^s_a/P^s_r$. 

\subsubsection{Puzzle Seed Release}\label{sec:seed_release}~\\
The puzzle seed release process requires the participation of at least $n$ of Tor's DAs. 
$n$ should include the majority of DAs to avoid centralization, and meanwhile 
it does not need to include all DAs to be fault-tolerant, similar to how the Tor network consensus is released. 
Specifically, each DA contributes its part for puzzle seed by generating a random nonce signed 
by its public key along with a timestamp, as formulated below.
\begin{equation}
s_i = n_i ~|~ t_s ~|~ \mathcal{S}_{D_i}, 
\end{equation}
where $n_i$ is the nonce, $t_s$ is the timestamp set to the starting time of the current $P^s_r$ (\eg $t_1$ in 
Figure~\ref{fig:puzzle_period}) to indicate the freshness of the $s_i$, and $\mathcal{S}_{D_i}$ is the $i$th DA's 
signature to prove the integrity of $s_i$.  To construct a puzzle seed for the current $P^s_r$, clients need to 
concatenate at least $n$ authentic seed pieces issued by $n$ distinct DAs. We note that Tor's existing random 
value generator~\cite{Tor_random_number} does not fit for \sys since it computes a fresh value every day 
whereas \sys's puzzle seeds need to be released more frequently to improve usability, as explained in \S~\ref{sec:puzzle_acceptance}. 

\subsubsection{Puzzle Computation}\label{sec:puzzle_generation}~\\
Assume that in the $k$th $P^s_r$, the puzzle seed obtained by a client is $h_k$. Then one puzzle is computed as follows. 
\begin{equation}\label{equ:puzzle}
p = H(p_{stub}) = H(h_k ~|~ r ~|~ s ~|~ \mathcal{F}), 
\end{equation}
where $H()$ is a public cryptographic hash function (\eg SHA512), $r$ is a random 128-bit cryptographic nonce generated by the client to 
ensure the uniqueness of the puzzle, $s$ is a 128-bit solution to the puzzle $p$, and $\mathcal{F}$ is the fingerprint of the 
AA selected by the client to redeem the puzzle solution. $s$ is considered as a valid solution to $p$ only if $\frac{p}{2^{L_0} - 1} < p_p$, 
where $L_0$ is the length of the hash function's output. $p_p$ is a parameter for tuning the puzzle system. We provide detailed 
discussion for $p_p$ in \S~\ref{sec:puzzle_system_analysis}. The $\{h_k~|~r~|~s~|~\mathcal{F}\}$ is defined as the puzzle stub $p_{stub}$, 
which will be sent to the AA (specified by $\mathcal{F}$) to serve as a capability seed. Incorporating $\mathcal{F}$ into the puzzle 
design prevents the client from redeeming a single puzzle solution at multiple AAs. 

\subsubsection{Puzzle Solution Acceptance Period}\label{sec:puzzle_acceptance}~\\
In order to be treated as valid capability seeds, puzzle stubs must be returned to the AAs within the puzzle solution acceptation period 
$P^s_a$, \ie an AA only accepts puzzle solutions received within  $[t_1, t_2]$ in the current $P^s_r$ (and equivalently $[t_3, t_4]$ in the next 
$P^s_r$), as illustrated in Figure~\ref{fig:puzzle_period}. The slot $[t_2, t_3]$ is the \emph{cool-down} period, during which no 
puzzle stubs are accepted. As a result, a client, regardless of whether it is bot or a legitimate Tor user, can spend at most  
$\frac{P^s_a-P_c}{P^s_r} < \frac{P^s_a}{P^s_r}$ percent of its CPU cycles on solving computational puzzles, where $P_{c}$ is the 
networking latency for retrieving the puzzle seed and returning the puzzle stub to the AAs. 

Because of the cool-down period in each $P^s_r$, clients who missed the current $P^s_a$ (either because they do not 
compute valid solutions on time or they obtain the puzzle seed later than $t_2$) will have to wait until the starting of the 
next $P^s_r$ (\eg $t_3$) to get another chance to solve new puzzles. As a result, $P^s_r$ needs to be small (\eg at most few 
minutes) to avoid introducing usability problems.

\subsubsection{Puzzle Solution Verification}\label{sec:puzzle_verification}~\\
To initiate puzzle verfication, a client sends the puzzle stub the corresponding AA. Upon the reception of puzzle stub, the AA 
performs the following checks to validate the puzzle stub. \first The puzzle stub is returned within the current $P^s_a$. \second The puzzle 
stub is computed based on the fresh puzzle seed released in the current $P^s_r$. \third The puzzle stub encloses its own fingerprint. 
\four The puzzle solution is valid, as defined in \S \ref{sec:puzzle_generation}. \five The puzzle stub has not been spent before. 
To enforce the fifth rule, the AA needs to cache all spent puzzle stubs. The cache space is bounded as the AA can erase 
the puzzle stubs received in previous periods since they are no longer spendable. 

\subsubsection{Puzzle System Analysis}\label{sec:puzzle_system_analysis}~\\
In each $P^s_r$, the number of puzzles solved by a  client follows the following binomial distribution
\begin{equation}\label{equ:A_seed}
\mathcal{G}_{seed} \sim \mathcal{B} \bigg( p_p, \floor*{\frac{P^s_a - P_c}{t_p}} \bigg),
\end{equation}
where $\mathcal{G}_{seed}$ denotes number of solved puzzles, $p_p$ is the probability that one attempt (\ie a hash computation)
produces a valid puzzle solution according to the rule in \S~\ref{sec:puzzle_generation}, $t_p$ is the amount of time it takes for 
the client to attempt a single hash computation and $\floor*{\frac{P^s_a - P_c}{t_p}}$ is the number of attempts $\mathbb{U}$ 
can make within the allowed time period.

\sys can control $p_p$ and $P^s_a$ to affect the numeric values of $\mathcal{G}_{seed}$. In particular, $p_p$ should be chosen
such that with high probability (\ie 0.99) a client with slow computation speed (\eg a mobile device released few years ago) 
and slow network connection (\eg 99th percentile of the RTT measured by CAIDA~\cite{RTT_dataset}) can correctly solve one 
puzzle so as to produce a valid capability seed. In particular, given that the slow device's computation speed is $t_p^0$ and 
the 99th percentile network latency is $P_c^{99th}$, $p_p$ is selected such that $1-(1-p_p)^{N_0} > 0.99$, where 
$N_0 = \floor*{\frac{P^s_a - P_c^{99th}}{t_p^0}}$.

\subsection{Trans-Capability Design Design}\label{sec:appendix:trans_capability}
In this section, we detail the capability exchange protocol discussed in \S~\ref{sec:capability_trans}. 
Since a Tor orion server (itself runs a Tor client) needs to open many Tor 
circuits in order to serve all its clients (\ie OS-clients) via \sys-enhanced circuits,  
enforcing per-seed rate limiting for pre-capability release may limit the availability of 
Tor OSes. To address this issue, we design the following capability exchange protocol. 

In particular, a OS-client needs to request a new type of capability, \ie \emph{trans-capability}, from the AAs. 
During the hidden service set up process, along with the information about Rendezvous Point, the OS-client 
sends a trans-capability to one of the OS's Introduction Points. The OS subsequently 
redeems the trans-capability at the AAs for new pre-capabilities, which can be used for generating new relay-specific 
capabilities. The trans-capability, accounted on the capability seed of the OS-client, anonymously informs the AAs that 
the OS needs to create a new circuit to serve the OS-client. 

\parab{Trans-capability Computation.}
Each trans-capability is computed based on \emph{pre-trans} issued by the AAs. By default, all Tor clients request pre-trans  
to prevent the AAs from knowing whether a client has the intention to visit OSes. Clients that do not visit any OS simply ignore 
the received pre-trans. To request pre-trans, the OS-client sends $\{\varsigma ~|~ n ~|~ t_s\}^b$ to the AAs, where 
$\varsigma$ is a pre-defined system value  for trans-capability. No information about the OS is enclosed to protect the OS's privacy. 
The AAs then compute blind signatures over the information to produce a pre-trans. Finally, the OS-client unblinds the received 
pre-trans to produce a trans-capability.

\parab{Redeeming Trans-capability.}  
The process of redeeming trans-capability is identical to how a Tor client requests relay-specific pre-capabilities  
using its capability seed (the trans-capability now serves as a new capability seed). The AAs reject all unauthentic, 
expired or spent trans-capabilities.

\subsection{Live Tor Interaction}\label{sec:appendix:live_tor}
In this section, we continue our discussion in \S~\ref{sec:implementation_tor} for live Tor network interactive. 

\begin{figure}[h]
	\centering
	\mbox{
		\subfigure{\includegraphics[scale=0.43]
			{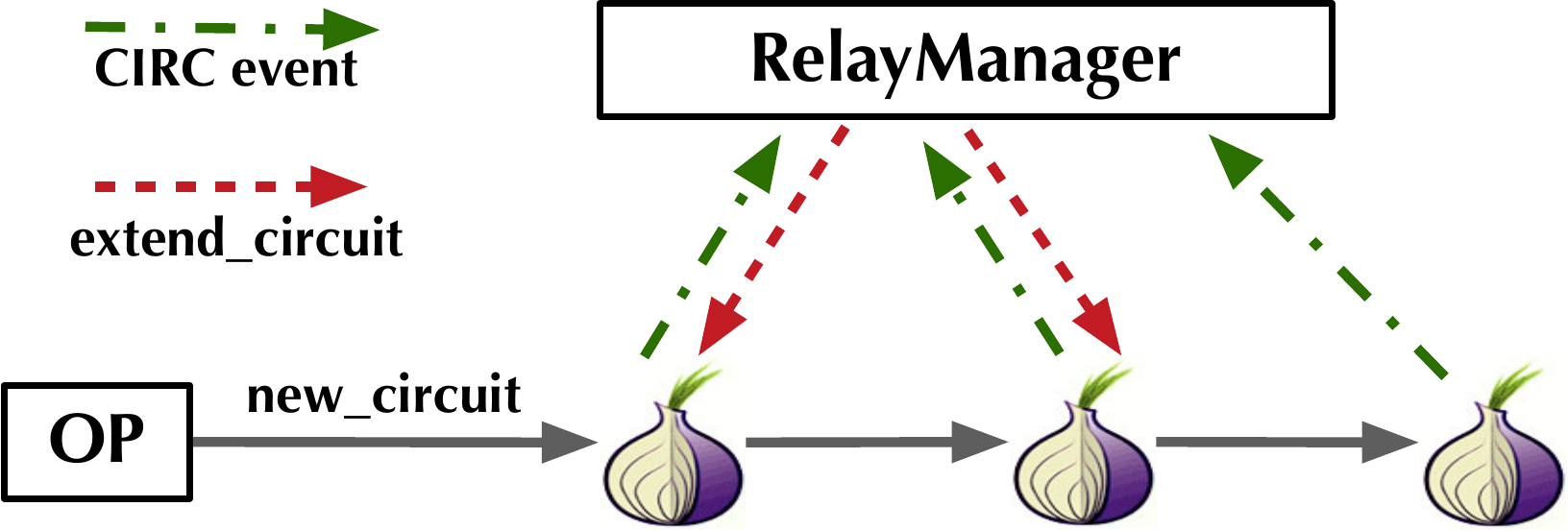}}
	}
	\caption{The design of \textsf{RelayManager}.
	}
	\label{fig:relaymanager}
\end{figure}

\parab{\textsf{RelayManager} Design.}
Since the live Tor relays are capability-agnostic (\ie they do not run our modified Tor source code described 
in \S~\ref{sec:implementation_tor}), we cannot create \sys-enhanced Tor circuits directly through live Tor relays. Thus, we implement 
another prototype to interact with live Tor relays during capability-enhanced circuit creation. The prototype relies on 
the Tor control protocol~\cite{tor_control}. In particular, on the OP, we implement a \textsf{RelayManager} based on the Stem~\cite{stem} library   
to execute the capability-related operations, as illustrated in Figure~\ref{fig:relaymanager}. The \textsf{RelayManager} controls the OP's 
circuit creation to ``embed'' capabilities into live Tor circuit creations.  In particular, after the OP selects relays for its circuit. 
the \textsf{RelayManager} blinds relay information, requests pre-capabilities from our deployed AAs 
described in \S~\ref{sec:implementation:AAs}, and then computes relay-specific capabilities. 
Whenever the OP's partially-built circuit reaches a relay $\mathbb{R}_n$ in the live Tor network, 
the \textsf{RelayManager} receives a \textsf{CIRC} event callback from the Tor control protocol.  As $\mathbb{R}_n$ is 
capability-agnostic, we offload  capability verification to the \textsf{RelayManager}. Upon validation, the \textsf{RelayManager} sends an 
\textsf{extend\_circuit} command through the Tor control protocol to continue circuit creation. Otherwise, the \textsf{RelayManager}  
terminates circuit creation by issuing a \textsf{close\_circuit} command. 

We clarify that \textsf{RelayManager} should not be used in real-world deployment since the capability 
verification is offloaded to the OP. Rather, to securely embrace \sys, the Tor relay source code needs to be modified properly, 
as proposed in \S~\ref{sec:implementation_tor}.

\begin{figure}[t]
	\centering
	\mbox{
		\subfigure{\includegraphics[scale=0.45]
			{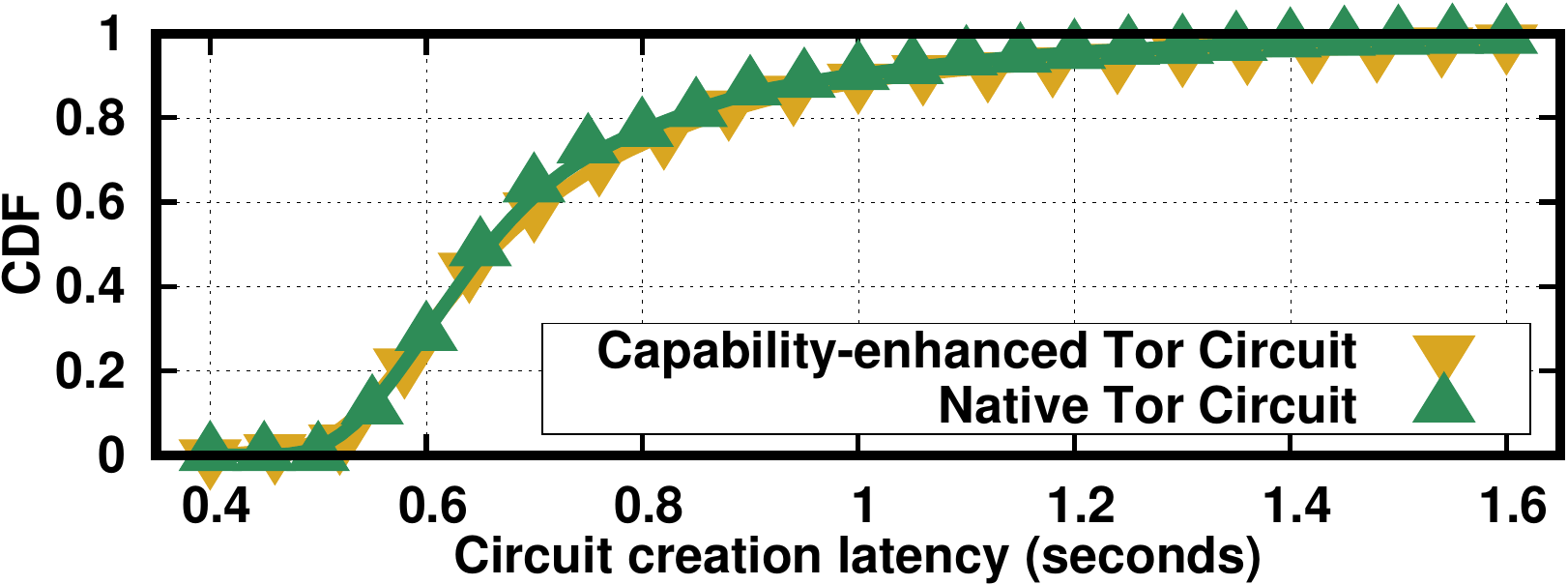}}
	}	
	\caption{\sys introduces negligible overhead during capability-enhanced Tor circuit creation.}
	\label{fig:live_tor}
\end{figure}		

\parab{Latency Measurement.} To validate the design of \textsf{RelayManager}, we instruct our Tor clients to create circuits 
through live Tor relays and use \textsf{RelayManager} to embed our capability-related operations into the creation process.  
Meanwhile, we measure circuit creation latencies to quantify the overhead caused by the our capability design.  
Figure~\ref{fig:live_tor} plots the CDF of measured latency with and without relay-specific capabilities. The results show that \sys 
introduces negligible overhead. This is because  a capability verification operation merely takes ${\sim}0.03$ milliseconds 
(\S~\ref{sec:evaluation_capability_overhead}) whereas the median time for creating a (native) Tor circuit is ${\sim}0.7$s. 
We clarify that since the capability verification is offloaded from live Tor relays to the OP, there could be some marginal errors 
in our latency measurements since live Tor relays may simultaneously process multiple circuits requests.

\subsection{Enforcing Site-Defined Policies}\label{sec:appendix:bound}
In this section, we continue the analysis in \S~\ref{sec:eva_site_policies} to prove that an adversary's service request rate $r_a$ 
is always bounded by $\Theta(\epsilon)$. In particular, when $k \leq 1$, CAPTCHA solutions are the optimal seeds since $c_0' < c_1'$. 
Therefore, the optimal $r_a$ is obtained when the adversary spends all investment on purchasing CAPTCHA solutions. 
Thus we have $r_a = \epsilon$. When $k \ge 1$, solutions to computational puzzles become the optimal seeds. In this case, 
the adversary's optimal $r_a$ is $k \cdot \epsilon$ which is obtained by investing all money on solving computational puzzles. 

When the site adopts the basic strategy (\ie accepts all Tor-emitted requests with valid capabilities), the adversary can continuously use  
optimal capability seeds to maintain its optimal $r_a$ as its investment increases. However, if either the rate limiting strategy or the 
WFQ strategy is adopted, the adversary will reach a point of diminishing returns when the site no longer accepts service requests 
using capabilities that are obtained via the optimal seeds. As a result, the adversary's $r_a$ starts to drop from the optimal value. 

Thus, regardless of $k$ and the site's strategy, $r_a$ is always bounded by $\Theta(\epsilon)$. The above analysis can be easily extended 
to more types of capability seeds. 

\subsection{Modeling for Botnet C\&C Abuse}\label{sec:appendix:modeling}
In this section, we continue the discussion in \S~\ref{sec:evaluation:CC} to detail the mathematical modeling for estimating the 
amount of circuit creation requests in Tor when Tor was under the large scale C\&C abuse happened in September 2013.

We collect the Tor network consensus published from September 1 to September 30, 2013 when the number of estimated daily Tor users 
ranged from 4 million to 6 million. Since one consensus file is published in each hour, we use the average statistics from all 24 
consensus files published in a day to represent the Tor status in that day. We model the relay computation capacity based on 
the live Tor relay measurements in~\cite{cellflood} by uniformly sampling their measurement numbers, excluding 
the samples with low confidence (as defined in their paper). 

The number of circuit creation requests received by Tor is modeled by a Poisson Process with arrival rate $\lambda$. To compute  
$\lambda$, we first estimate the number of unique Tor clients in a time interval and then estimate the number of circuits opened 
by each client in the same interval. Mathematically, we have $\lambda = \frac{N_1 \cdot r_1 + N_2 \cdot r_2}{t_0}$, where $N_1$ and 
$N_2$ are the number of unique legitimate clients and bot clients, respectively, over the time interval $t_0$; $r_1$ and $r_2$ are the 
average number of circuit creations requested by a legitimate client and a bot client, respectively, over the same interval $t_0$. $N_1$ 
and $N_2$ can be estimated using the metric inferred from live Tor measurements in~\cite{PrivCount}. In particular, 
over a 10-minute interval, PrivCount~\cite{PrivCount} counts 710 unique clients when Tor's daily estimated user 
count is 1.75 million, which indicates the client population  \emph{turnover rate} $\rho$ is about 2.5. Since the methodology 
used by Tor to estimate its daily user has not changed since 2013, we assume that $\rho$ obtained in 2016 is also applicable in 2013. 
Thus, over a 10-minute interval, we have $N_1 = N_1^L / \rho$ and $N_2 = (N_2^T - N_1^L)/\rho$, 
where $N_1^L$ and $N_2^T$ are the number of legitimate daily users and total daily users estimated by Tor, respectively. 
We estimate $N_1^L$ as 1 million, which was the estimated daily Tor user number right before the abuse started 
in August 2013. $N_2^T$ can obtained directly from the data published by Tor~\cite{tor_metric}. Further, PrivCount~\cite{PrivCount} counts  
about $4$ circuits opened for each Tor client over a 10-minute interval,  thus we estimate $r_1$ is about $4$ in a 10-minute interval, assuming 
that legitimate Tor clients had the same usage pattern in 2013 as they have in 2016. 

However, the above usage pattern inferred from \cite{PrivCount} cannot be applied to determine $r_2$ since bot clients may 
have different usage patterns from legitimate clients. Thus, we estimate $r_2$ using historical data. In particular, we find that the highest 
circuit creation failure rate on September 27 2013 is about 35\%~\cite{hidden_service_abuse}. Then using the network consensus of the same day, 
$r_2$ is estimated at about $150$ over a 10-minute interval.

Based on the above mathematical modeling, we study the circuit creation failure rates using our Tor-scale simulator. 
The results are plotted in Figure~\ref{fig:circuit_failure_rate}. 

\end{document}